\documentclass[a4paper,twocolumn,11pt,accepted=2025-08-12]{quantumarticle}
\pdfoutput=1

\usepackage[english]{babel}
\usepackage[T1]{fontenc}
\usepackage{amsmath}
\usepackage{amssymb}
\usepackage{hyperref}
\usepackage{tikz}
\usepackage[utf8]{inputenc}
\usepackage{float}
\usepackage{algorithm}
\usepackage{algorithmic}
\usepackage{svg}
\usepackage{tikz}
\usetikzlibrary{quantikz}
\usepackage{braket}
\usepackage{lipsum}
\usepackage{graphicx}% Include figure files
\usepackage{dcolumn}% Align table columns on decimal point
\usepackage{bm}% bold math
\usepackage{physics}
\usepackage{url}
\usepackage{xcolor}
\usepackage{bbm}
\usepackage{footnote}
\usepackage{amsthm}
\usepackage{xcolor}
\theoremstyle{plain}
\newtheorem{theorem}{Theorem}

\newtheorem{proposition}[theorem]{Proposition}
\newtheorem{conjecture}[theorem]{Conjecture}

\newcommand{\Mod}[1]{\ (\mathrm{mod}\ #1)}

\definecolor{crimson}{rgb}{.8, 0, 0}

\begin{document}

\title{Preparing low-variance states using a distributed quantum algorithm}
\author{Xiaoyu~Liu}
\email{xiaoyu@lorentz.leidenuniv.nl}
\affiliation{$\langle aQa ^L\rangle $ Applied Quantum Algorithms, Universiteit Leiden}
\affiliation{Instituut-Lorentz, Universiteit Leiden, P.O. Box 9506, 2300 RA Leiden, The Netherlands}
\author{Benjamin~F.~Schiffer}
\email{benjamin.schiffer@mpq.mpg.de}
\affiliation{Max-Planck-Institut f\"ur Quantenoptik, Hans-Kopfermann-Str.~1, D-85748~Garching, Germany}
\author{Jordi~Tura}
\affiliation{$\langle aQa ^L\rangle $ Applied Quantum Algorithms, Universiteit Leiden}
\affiliation{Instituut-Lorentz, Universiteit Leiden, P.O. Box 9506, 2300 RA Leiden, The Netherlands}

\begin{abstract}
Quantum computers are a highly promising tool for efficiently simulating quantum many-body systems. 
The preparation of their eigenstates is of particular interest and can be addressed, e.g., by quantum phase estimation algorithms. The routine then acts as an effective filtering operation, reducing the energy variance of the initial state.
In this work, we present a distributed quantum algorithm inspired by iterative phase estimation to prepare low-variance states.
Our method uses a single auxiliary qubit per quantum device, which controls its dynamics, and a postselection strategy for a joint quantum measurement on such auxiliary qubits. 
In the multi-device case, the result of this measurement heralds the successful runs of the protocol.
This allows us to demonstrate that our distributed algorithm reduces the energy variance faster compared to single-device implementations, thereby highlighting the potential of distributed algorithms for near-term and early fault-tolerant devices. 
\end{abstract}

\section{Introduction}

Preparing eigenstates of Hamiltonians plays an essential role in quantum simulation to investigate the properties of quantum many-body systems.
Quantum computers are naturally well-suited for this task, with quantum phase estimation being the canonical method for eigenstate preparation~\cite{nielsen2010quantum,kitaev1995quantum,obrien2019quantum,dutkiewicz2022heisenberglimited}. Obtaining precise eigenstates then allows one to accurately compute their energies and measure other observables that provide insights into the behavior of complex many-body systems.

\begin{figure*}[t]
    \centering
    \includegraphics[width=1.0\linewidth]{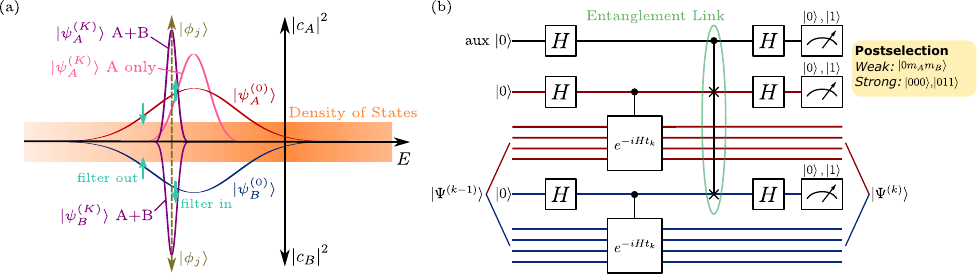}
    \caption{(a) Illustration of a filtering operation. Consider a local Hamiltonian $H$, which has an approximate Gaussian-shaped density of states centered at $\Tr(H)$~\cite{hartmann2005spectral}. The devices Alice ($A$) and Bob ($B$) start by preparing an initial state, with respective populations $\left|c_A\right|^2$ and $\left|c_B\right|^2$. 
    After multiple iterations of the protocol, the state becomes increasingly concentrated. 
    The distributed version of this algorithm uses an entanglement link between $A$ and $B$ to enhance the filtering process, resulting in an output state with a lower energy variance on average for the same number of iterations. 
    Note that an energy bias may be introduced during the postselection routine.
    (b) Circuit for the distributed filtering algorithm. Both $A$ and $B$ are initialized in the same state $|\Psi^{(0)}\rangle=|\psi^{(0)}\rangle^{\otimes 2}$. In each iteration $k$, a random time evolution for a sufficiently large $t_k$ is applied to both controlled unitaries, followed by a control-SWAP gate, realizing the entanglement link. We analyze two cases (\emph{weak} and \emph{strong}) for postselecting the measurement results of the auxiliary qubits.}
    \label{intro}
\end{figure*}

Quantum phase estimation (QPE)~\cite{nielsen2010quantum,abrams1999quantum} and filtering~\cite{poulin2009preparing, ge2019faster} are closely related methods to prepare low-variance states on quantum computers. For a Hamiltonian of interest $H$, the algorithm relies on conditional dynamics $e^{-iHt}$ that realizes phase kick-back to the controlling auxiliary qubit.
Another interesting approach is to prepare the filtered state virtually by measuring Loschmidt echoes~\cite{lu2021algorithms}, or adiabatically by a parent Hamiltonian construction~\cite{irmejs2024efficient}.
All approaches have in common that the system is initialized in a state that can be efficiently prepared. Typically, this initial state is a product state that will generally exhibit a Gaussian-like distribution in the eigenbasis of a local Hamiltonian~\cite{rai2024matrix}. 
A filter algorithm acts on the initial state by sharpening the energy distribution towards a delta function, as shown in Fig.~\ref{intro}(a). 
In the limit of an arbitrarily precise filter centered at an eigenenergy of the Hamiltonian, the operation prepares the microcanonical ensemble by isolating the eigenstate(s) at that energy. 

A particularly simple quantum routine for filtering is inspired by iterative quantum phase estimation (IQPE)~\cite{abrams1999quantum, dobsicek2007arbitrary, xu2014demonlike, meister2022resourcefrugal, chen2020quantum, schiffer2025quantum, qian2024demonstration, choi2021rodeo}. This algorithm applies a Hadamard-test circuit~\cite{aharonov2006polynomial} for each iteration by implementing a time evolution operator controlled on a single auxiliary qubit.
The auxiliary qubit can then be reused for the next iteration.
Conversely, textbook QPE relies on a sequence of conditional dynamics (controlled unitary evolutions) with evolution times chosen such that one may obtain one additional digit of the phase with each auxiliary qubit. 
This precision requirement can be relaxed by choosing the evolution times randomly from a sufficiently large time interval to shuffle the eigenphases, as we explain further below.
For this protocol, it is known that the energy variance of the output state decreases with each iteration on average, generically converging to an eigenstate. However, reaching low-variance states may still require many iterations, demanding qubits with long coherence times.

A promising strategy to reduce the number of rounds --- and consequently, the circuit depth --- in this IQPE-based protocol is to leverage distributed quantum computing. Distributed quantum algorithms can address large-scale problems with shallower circuits by utilizing interconnected quantum devices that can share entanglement~\cite{magnard2020microwave,harvey-collard2022coherent,cirac1999distributed,buhrman2003distributed}. 
In fact, such a distributed approach has been applied to eigenstate preparation, as demonstrated in~\cite{schiffer2025quantum}. In that work, an eigenstate broadcasting scenario was investigated with two quantum devices: a first device Alice, with an almost perfect eigenstate, and a second device Bob, with a rougher approximation to the target ground state. The distributed circuit allows Bob to speed up his preparation via implicit knowledge of the target state through a shared entanglement link with Alice. 

In this work, we analyze a similar setup, focusing on its application to quantum state filtering:
both Alice and Bob start with the same product state and iteratively converge to states with a lower energy variance.
We first propose a two-device distributed filtering algorithm where joint quantum measurements between the devices enable Alice and Bob to control the moments of the energy distribution of the initial state. We focus on the first and second moments, i.e.~the energy and variance of the state. 
Concretely, we show that the distributed filtering algorithm reaches a low-variance state more rapidly on average than the single-device case.
Postselection on the measurement outcome is an essential ingredient in the protocol and careful engineering of the postselection allows one to further optimize the state preparation process. 
The postselection introduces either a positive or negative energy bias. 
We provide a bound on the maximum bias introduced in the protocol.
Finally, we generalize our algorithm to a larger number of connected quantum devices that jointly perform the filtering operation. Numerical evidence shows an enhanced eigenstate preparation as the number of devices is increased.  

This paper is structured as follows: in Section~\ref{sec2}, we introduce the filtering algorithm and provide a detailed explanation of the distributed algorithm, including its pseudocode. We analyze the performance of the two-device algorithm through numerical experiments on several product state instances, focusing on the behavior of the average variance and energy over iterations. The extension of the distributed architecture to a setup with more than two devices is found in Section~\ref{sec3}, including numerical experiments benchmarking the performance of two and three devices against a single device. Finally, in Section~\ref{sec4}, we summarize and discuss our findings.

\begin{figure*}[t]
    \centering
    \includegraphics[width=1.0\linewidth]{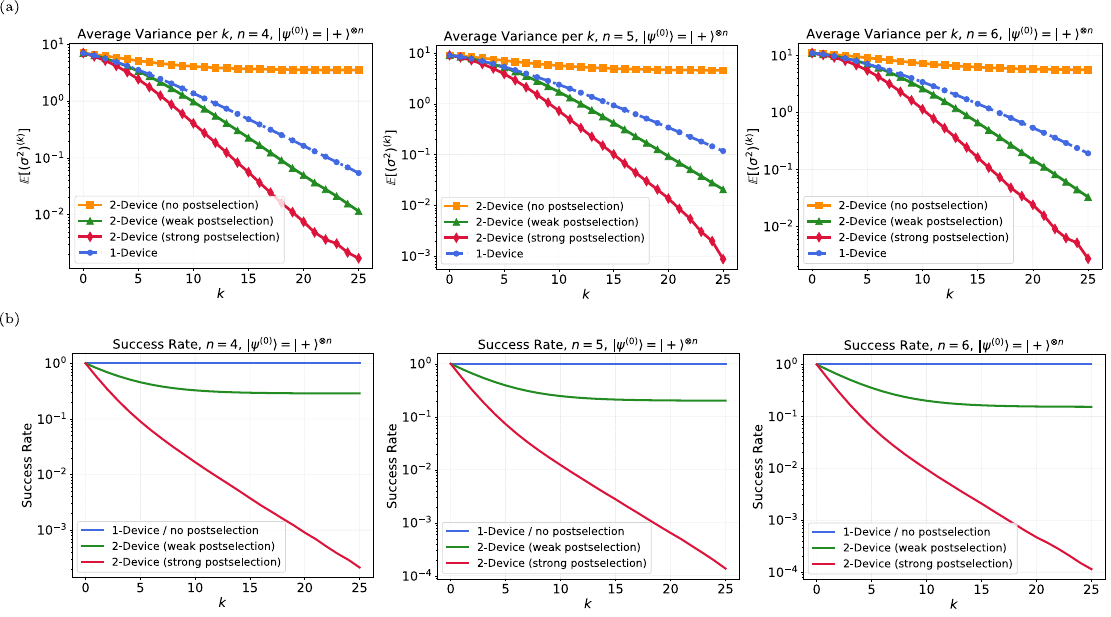}
    \caption{(a) Average variance of the pure state ensemble per iteration $k$ up to 25. Three cases ($n=4,5,6$, $\ket{\psi^{(0)}}=\ket{+}^{\otimes n}$) are shown. For cases requiring postselection, only states meeting the postselection criteria are retained at each iteration $k$. The cumulative success rate after postselection at each iteration are shown in (b). Note that the single-device algorithm does not require postselection, thus resulting in a 100\% success rate. We use $10^8$ trials from the beginning for \emph{strong} postselection cases and $10^6$ for the rest to acquire accurate enough expectation values. More numerical results starting with a fixed number of repetitions of the algorithm with complete error bars are shown in the Appendix~\ref{appx-complete}.}
    \label{var_psi}
\end{figure*}

\section{Distributed filtering with auxiliary qubits\label{sec2}}

The single-device setup was previously introduced under different names in~\cite{abrams1999quantum, xu2014demonlike, chen2020quantum, meister2022resourcefrugal, schiffer2025quantum}, where Hadamard tests~\cite{aharonov2006polynomial} are applied iteratively using the unitary evolution operator $U=e^{-iHt_k}$ for each iteration $k$. In each iteration, this unitary is controlled by a single local auxiliary qubit, which is initialized in $\ket{+}$ and measured in the $\ket{\pm}$ basis. The details of the single-device algorithm are restated in Appendix~\ref{appx-single}. 

To extend this setup to two devices, we use the configuration shown in Fig.~\ref{intro}(b). Here, two identical Hadamard test circuits are executed in parallel on both devices Alice and Bob. The two devices are connected through their respective control qubits to a single additional helper qubit ``aux''. This  allows quantum communication between the devices using a circuit that is effectively a swap test circuit on the two control qubits~\cite{barenco1997stabilization,buhrman2001quantum}, where the control-SWAP gate acts as the entanglement link, acting non-locally only on the spatially separated auxiliary qubits. For simplicity, we assume that Alice possesses the top auxiliary qubit (``aux'' in Fig.~\ref{intro}). Then, for each iteration, only one Bell pair generated between Alice and Bob is required to teleport Bob's single local auxiliary qubit and the control-SWAP gate can be applied locally. The three qubits are finally measured in the $\ket{\pm}$-basis. We detail the realization of the distributed circuit with shared Bell pairs between the two devices in Appendix \ref{appendix:cost}.

In the following sections, we present a detailed description of the distributed filtering algorithm and demonstrate its ability to control the moments of quantum states. 
We focus our analysis on the first and second moment --- energy $E$ and variance $\sigma^2$ --- defined as
\begin{align} \label{E_def_var_def}
    E &= \bra{\psi}H\ket{\psi}, \\
    \sigma^2 &= \bra{\psi}(H-E)^2\ket{\psi} = \bra{\psi}H^2\ket{\psi} - E^2,
\end{align}
for a pure state $\ket{\psi}$.
Note that for eigenstates of $H$ we have $\sigma^2=0$.

\subsection{Algorithm outline}

Recall that the objective of the filtering algorithm is to prepare a state with lower variance in the eigenbasis of a given Hamiltonian $H$ than the initial (product) state. 
In the distributed setup, both Alice and Bob start with identical product states: $\ket{\psi_A}\otimes\ket{\psi_B}=\ket{\psi^{(0)}}^{\otimes2}$. 
They each input their state locally into the circuit shown in Fig.~\ref{intro}(b) and run the circuit iteratively. 
At each iteration, two identical controlled unitaries $e^{-iHt_k}$ are applied locally. 
The time $t_k$ for each iteration is randomly chosen from $(0,T]$, such that $\varphi_j^{(k)} = -t_k \lambda_j \Mod{2\pi} \sim \text{Uniform}(0, 2\pi)$ for every $j$, to effectively separate the eigenstates. Here $\lambda_j$ is the $j$th eigenvalue of the Hamiltonian $H$.
This randomization is needed because, without prior knowledge of the initial product state's populations in the eigenbasis, each iteration requires a sufficiently large time evolution $t_k$ to effectively separate the eigenstates from each other by randomizing the eigenphases $\varphi^{(k)}_j=-t_k\lambda_j$. 
Note that both Alice and Bob use the same $t_k$ in each iteration to ensure that the output state remains symmetric across Alice and Bob. 
This condition is necessary for both parties to converge to the same eigenstate, heralded by postselection. 
This synchronization can be efficiently achieved by pre-sharing a list of random $t_k$ values generated from a common random seed.
Further details on this randomization of the phases are provided in Appendix~\ref{appx-single-var}.

Additionally, the distributed filtering algorithm depends on postselection of the measured auxiliary qubits, which can be performed using either a \emph{weak} postselection criterion (where we continue to the next iteration only when top auxiliary qubit outcome is $\ket{0}$), or a \emph{strong} criterion (where we continue to the next iteration only when the outcome of the three auxiliary qubits is $\ket{000}$ or $\ket{011}$).
These two types of postselection exhibit different properties and performance, which we explore later. 
If any iteration yields measurement outcomes that fall outside the postselection criteria, the protocol is restarted, as such outcomes on average broaden the state distribution, preventing the operation from effectively working as a filter.
The pseudocode for the distributed filtering algorithm is shown in Algorithm~\ref{algo1}.

For the \emph{strong} postselection, $\ket{\Psi_{AB}^{(K)}}$ will have a tensor product structure across Alice and Bob. If we start with two copies of the same state, both Alice and Bob will hold the same final state. On the other hand, for the \emph{weak} postselection, the state will be entangled in general, but for large values of $K$, it will approximate a product of identical states.
This is because both postselection criteria select the top qubit to be measured in the $\ket{0}$ state, the state registers of both Alice and Bob are projected to their joint symmetric subspace in every iteration.
As a result, tracing out either Alice or Bob's systems will result in the same states on both devices. %, and the final prepared eigenstates for both parties will be identical.
If the minimal distance between eigenvalues of the Hamiltonian $H$ is $\Omega(1/t_k)$, then the final states will be eigenstates of $H$.
Our distributed protocol therefore allows one to prepare multiple identical copies of (approximate) eigenstates in a heralded way, and we include the details in Appendix~\ref{appendix:energy-variance}.

\subsection{Faster reduction of the variance}

We now analyze the performance of the distributed filtering algorithm. The variance of the pure state is defined in Eq.~\eqref{E_def_var_def}. Due to the inherent randomness of the algorithm (random $t_k$ and the measurement process), the output state after a fixed number of iterations can differ. 
Consequently, we compute the average variance of the pure state ensemble at iteration $k$, denoted $\mathbb{E}\left[ (\sigma^{2})^{(k)} \right]$, over a sufficiently large number of repeated runs. 
Analytically, we write
\begin{align} \label{eq:variance}
    &\mathbb{E}\left[ (\sigma^{2})^{(k)} \right] \nonumber\\
    =&\mathbb{E}\left[\bra{\psi^{(k)}}H^2\ket{\psi^{(k)}}\right] - \mathbb{E}\left[\left( \bra{\psi^{(k)}}H\ket{\psi^{(k)}} \right)^2\right] \nonumber\\
    =&\Tr\left( \rho^{(k)}H^2 \right) - \mathbb{E}\left[\left( \bra{\psi^{(k)}}H\ket{\psi^{(k)}} \right)^2\right],
\end{align}
where $\ket{\psi^{(k)}}$ denotes the output state at iteration $k$, and $\rho^{(k)}$ represents the density matrix describing a mixture of states $\{\ket{\psi^{(k)}}\}$ with the corresponding probabilities. 

%\begin{figure}
\begin{algorithm}[H]
\label{algo-two}
    \caption{Two-device distributed filtering algorithm}
    \begin{algorithmic}[1]
        \REQUIRE Two product states $\ket{\psi}\otimes\ket{\psi}=\ket{\psi^{(0)}}^{\otimes 2}$; target iteration $K$; postselection type $\mathcal{P}$; Hamiltonian $H$ with a generic spectrum
        \ENSURE State $\ket{\Psi^{(K)}_{AB}}$.
        \STATE Initialize: $\ket{\Psi_{AB}^{(0)}} \gets \ket{\psi^{(0)}}^{\otimes 2}$;
        \FOR{$k = 1:K$}
            \STATE Generate random time $t_k$ such that $\varphi_j^{(k)} = -t_k \lambda_j \mod 2\pi \sim \text{Uniform}(0, 2\pi)$ for every $j$;
            \STATE Input $\ket{\Psi_{AB}^{(k-1)}}$ and $t_k$ to the circuit in Fig.~\ref{intro}(b);
            \STATE Measure and get the result $\ket{m_k} = \ket{m^{(k)}_0 m^{(k)}_A m^{(k)}_B}$;
            \IF{$\mathcal{P} = \text{\emph{weak}} \textbf{ and } m^{(k)}_0 = 0$}
                \STATE Output $\ket{\Psi_{AB}^{(k)}}$ and \textbf{continue};
            \ELSIF{$\mathcal{P} = \text{\emph{strong}} \textbf{ and } (m_k = 000 \textbf{ or } m_k = 011)$}
                \STATE Output $\ket{\Psi_{AB}^{(k)}}$ and \textbf{continue};
            \ELSE
                \STATE \textbf{break} and start over from initialization;
            \ENDIF
        \ENDFOR
        \RETURN $\ket{\Psi_{AB}^{(K)}}$.
    \end{algorithmic}
    \label{algo1}
\end{algorithm}
%\end{figure}

We numerically compute the variance of the pure state ensemble and consider the $n$-qubit initial input state $\ket{\psi^{(0)}}=\sum_{j}c_j^{(0)}\ket{\phi_j}$, where the quantities $\ket{\phi_j}$ and $c^{(0)}_j$ denote the $j$th eigenstate of the Hamiltonian $H$ and its corresponding amplitude, respectively. 
Throughout this paper, we shall use the non-integrable, one-dimensional Ising model with both a longitudinal and a transverse field as a running example:
\begin{equation}
    H = \sum_{j=1}^{n-1}\sigma_j^z \sigma_{j+1}^z + \sum_{j=1}^n (\sigma_j^x+\sigma_j^z).
\end{equation}
In Fig.~\ref{var_psi} we present numerical results for $n=4,5,6$ and $\ket{\psi^{(0)}}=\ket{+}^{\otimes n}$. 
We compare the performance of the distributed filter with the single-device filter. 
For scenarios requiring postselection, states that do not meet the postselection criterion are discarded. 
The results show that the distributed filter with postselection achieves a lower average energy variance, outperforming the single-device filter.
In contrast, the distributed filter without postselection, where the algorithm continues regardless of measurement outcomes, performs worse than the single-device case. The \emph{strong} criterion for postselection yields the fastest decrease in the average energy variance.

Nevertheless, the distributed filter incurs an overhead due to postselection. Fig.~\ref{var_psi}(b) illustrates how this overhead scales with the number of iterations $k$ for both \emph{weak} and \emph{strong} postselection types. 
We calculate the cumulative success rate after each postselection, i.e., the average probability of the auxiliary qubit(s) result meeting the corresponding postselection criterion at each $k$. 
Notably, in the case of \emph{weak} postselection, the cumulative success rate tends to stabilize after a few iterations. This implies that once the protocol succeeds in several consecutive rounds, it is highly likely to continue succeeding.
Specifically, the average cumulative success rate of the \emph{weak} postselection is lower bounded:
\begin{align}
    P_{\mathrm{w}}&=\sum_{j}|c_j^{(0)}|^4+\left(\frac{3}{4}\right)^k\sum_{j\neq j'}|c_j^{(0)}|^2 |c^{(0)}_{j'}|^2 \nonumber\\
    &\geq\sum_{j}|c_j^{(0)}|^4.
\end{align}
Although this lower bound depends only on the the eigenbasis amplitudes of $\ket{\psi^{(0)}}$ which in principle can be any state, for $n$-qubit product state this lower bound indeed has a scaling of the form $A\exp(B/n)$ with $A$ and $B$ coefficients independent of $n$, when the Hamiltonian $H$ is local.
This lower bound decreases mildly with increasing $n$, which supports the practicality of the \emph{weak} postselection protocol for larger system sizes.
For \emph{strong} postselection, however, this quantity decreases exponentially as:
\begin{equation}
    P_{\mathrm{s}}=\left( \frac{3}{4} \right)^k\sum_{j}|c_j^{(0)}|^4 + \left( \frac{1}{2} \right)^k \sum_{j \neq j'}|c_j^{(0)}|^2|c^{(0)}_{j'}|^2.
\end{equation}
Both derivations are included in Appendix~\ref{appendix:energy-variance}.
Due to the overhead introduced by postselection, additional resource costs will be incurred, such as an increased number of time evolutions and greater consumption of Bell pairs.
These costs can be evaluated from the cumulative success rate, and for \emph{weak} postselection they scale linearly with $k$ when the cumulative success rate plateaus, but the ones for \emph{strong} postselection scale exponentially.
Thus, there is a trade-off between the two postselection types: \emph{strong} postselection reduces the energy variance more quickly than \emph{weak} postselection but suffers from an exponential postselection overhead.

\subsection{Spreading of eigenstates}

In the previous section, we considered the average energy variance of the pure state prepared in the algorithms, which asymptotically goes to zero in the limit of many iterations.
The probability with which we converge to a specific eigenstate, though, is related to the initial populations $\{|c^{(0)}_j|^2\}$. This is another variance and can be interpreted as the variance of the mixed state $ \rho^{(k)}$ that describes the expected physical system after $k$ iterations. To avoid confusion between the two types of variances, we denote it as the \emph{eigenstate spread}. Effectively, it is a measure over what energy range the eigenstates prepared in the algorithm are distributed. 
We can also see the difference between the variance of the pure states and the eigenstate spread, by looking at the analytic description of the eigenstate spread
\begin{align}
    V^{(k)} &= \Tr\left( \rho^{(k)}H^2 \right) - \mathbb{E}\left[ {E^{(k)}}^2 \right] \nonumber\\
    &= \Tr\left( \rho^{(k)}H^2 \right) - \left(\Tr(\rho^{(k)} H)\right)^2,
\end{align}
where the second term differs from Eq.~\eqref{eq:variance}.

We show the eigenvalue spread in Fig.~\ref{V_approx} and observe that the over several iterations of the protocol the spread decreases. 
For both of the postselection scenarios, the eigenvalue spread converges to a fixed value:
\begin{align} \label{eq:spread:general}
     V^{(k)}_{k\rightarrow\infty} = \frac{\sum_{j}\lambda_j^2 |c^{(0)}_j|^4}{\sum_{j}|c^{(0)}_j|^4}-\left( \frac{\sum_{j}\lambda_j |c^{(0)}_j|^4}{\sum_{j}|c^{(0)}_j|^4} \right)^2. 
\end{align}
This corresponds to the smallest populations having been suppressed. 
We observe the eigenvalue spread converge faster towards this value for the \emph{strong} postselection criterion.

\begin{figure}[t]
    \centering
    \includegraphics[width=.8\linewidth]{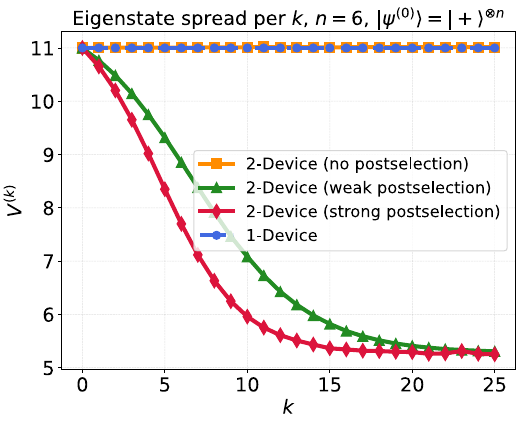}
    \caption{Eigenvalue spread $V^{(k)}$ for $n=6$ and $\ket{\psi^{(0)}}=\ket{+}^{\otimes n}$. Note that the eigenvalue spread goes to a fixed lower bound under the postselection circumstances.}
    \label{V_approx}
\end{figure}

We can find an analytical expression for the asymptotic behavior of the eigenstate spread by considering a smooth distribution for the amplitudes of the Hamiltonian eigenvalues and also for the initial populations of the state. 
Concretely, we assume $H$ is traceless, local and has a Gaussian density of states centered at zero with variance $\sigma^2$ (cf.~\cite{hartmann2005spectral}).
The input product state $\ket{\psi^{(0)}}$ is assumed to have eigenstate populations with Gaussian shape centered at energy $\mu$ and with variance $\xi^2$ (cf.~\cite{rai2024matrix}). 
The eigenstate spread is a function of the initial distributions, the difference in the limit of many iterations $k\rightarrow\infty$ from the initial eigenvalue spread is described by the expression
\begin{align} \label{eq:spread:cont}
    V^{(0)}-V^{(k)}_{k\rightarrow\infty} &=\frac{1}{\left( \frac{\xi^2}{\sigma^2}+1 \right)\left( \frac{\xi^2}{\sigma^2}+2 \right)} \xi^2.
\end{align}
In the limit of an extremely narrow energy distribution of the eigenstates ($\xi\rightarrow0$), the eigenvalue spread does not change anymore during the protocol.
We provide a rigorous derivation of the expressions Eq.~\eqref{eq:spread:general} and Eq.~\eqref{eq:spread:cont} in the Appendix~\ref{app:ssec:spread} for the \emph{weak} postselection criterion. For the \emph{strong} postselection criterion, we numerically observe that the same expression holds.

We remark that the ensemble $\rho^{(k)}$ is special in the sense that we know, in the large number of rounds limit, it is generated as a mixture of eigenstates of the Hamiltonian, but in a smaller energy window than that of the initial product state. 
The smaller window is a structural property that depends on the density of states and the initial state distribution.

\subsection{Energy bias analysis}

We now turn to the energy behavior of the resulting states, focusing on the ensemble average of the pure state energy at iteration $k$, given by
\begin{equation}
    \mathbb{E}\left[ E^{(k)} \right] = \mathbb{E}\left[ \bra{\psi^{(k)}}H\ket{\psi^{(k)}} \right] = \Tr\left( \rho^{(k)}H \right).
\end{equation}
Using the same setup as above, we have shown the average energy obtained in numerical simulations in Fig.~\ref{energy}. It follows from the Born rule~\cite{born1926zur}, that for both the single-device and two-device algorithms without postselection, the average energy remains constant, matching the energy of the input state $\ket{\psi^{(0)}}$:
\begin{equation}
    \mathbb{E}\left[ E^{(k)} \right] = \sum_{j} \lambda_j \left| c_j^{(0)} \right|^2, 
\end{equation}
where $\lambda_j$ is the energy of the $j$th eigenstate of $H$.
However, for the two distributed cases with postselection, energy biases are introduced in this process. 
Specifically, with the postselection, one can show that in the limit of many iterations $k$
\begin{align}
 \mathbb{E}\left[ E^{(k)} \right]_{k\rightarrow \infty} \rightarrow \bigg(\sum_j \lambda_j \left| c_j^{(0)} \right|^4\bigg) \Big/ \bigg(\sum_j \left| c_j^{(0)} \right|^4\bigg), 
\end{align}
effectively bounding the energy bias introduced.
Hence, the eigenenergies of the Hamiltonian and the initial state distribution fully determine the expected converged energy.
The effective effect can either drive the system towards an effective higher or lower average energy.
Again, similar to the energy spread, $\mathbb{E}\left[ E^{(k)} \right]$ with the \emph{strong} postselection criterion shows faster convergence.

\begin{figure}[t]
    \centering
    \includegraphics[width=.85\linewidth]{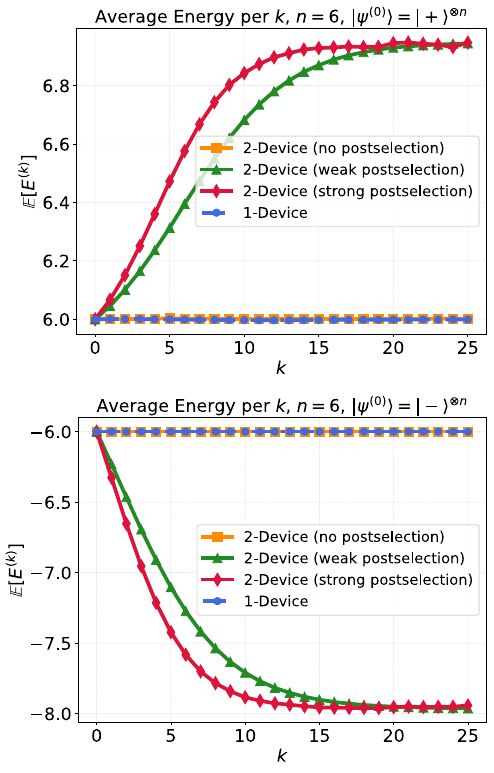}
    \caption{Average energy of the pure state ensemble per iteration $k$ for $n=6$, $\ket{\psi^{(0)}}=\ket{+}^{\otimes n}$ and $\ket{-}^{\otimes n}$. Energy biases are introduced in cases with postselection. Note that the energy bias can be either positive or negative, depending on the sign of $\mathbb{E}[E^{(k)}]_{k\rightarrow \infty}$. Additional numerical simulations with complete error bars are shown in the Appendix~\ref{appx-complete}.}
    \label{energy}
\end{figure}

We provide an additional analysis of the energy bias for the continuous case in Appendix~\ref{appx-dist-energy}, where both the density of states of the local $H$ and the populations of $\ket{\psi^{(0)}}$ have Gaussian shapes.
Again, we assume $H$ is traceless, local and has a Gaussian density of states centered at zero with variance $\sigma^2$. As in the previous section, the input product state $\ket{\psi^{(0)}}$ is assumed to have eigenstate populations with Gaussian shape centered at energy $\mu$ and with variance $\xi^2$. 
Then, we show that the largest bias we can introduce is
\begin{align}
    \left|E^{(0)}-\mathbb{E}\left[ E^{(k)} \right]_{k\rightarrow\infty}\right| &=\left|\frac{\mu}{\frac{\xi^2}{2\sigma^2}+1} -\frac{\mu}{\frac{\xi^2}{\sigma^2}+1}\right| \nonumber\\
    &= |\mu|\left( \frac{1}{\frac{\xi^2}{2\sigma^2}+1} -\frac{1}{\frac{\xi^2}{\sigma^2}+1} \right)\nonumber\\
    &\leq |\mu| \frac{1}{2\sqrt{2}+3} 
\end{align}
The full derivation is included in Appendix~\ref{app:ssec:weak}.

\section{Multi-device extension\label{sec3}}

We now show how to extend the distributed filter method to an arbitrary number of devices. 
A natural approach is to generalize the swap operator $S$ to a derangement operator $D$. One way to achieve this is shown by the following circuit:

\begin{equation}
\small
    \begin{quantikz}[row sep = {0.5cm,between origins}]
    \lstick{1} & \gate[5]{D} &\qw \rstick{2}\\
    \lstick{2} &  &\qw \rstick{3}\\
    \lstick{3} &  &\qw \rstick{4}\\
    \lstick{\vdots} & &\qw  \rstick{\vdots}\\
    \lstick{$s$} & &\qw \rstick{1}\\
    \end{quantikz}=
    \begin{quantikz}[row sep = {0.5cm,between origins}]
    \lstick{1} & \swap{1} &\qw &\qw &\qw &\qw \rstick{2} \\
    \lstick{2} & \targX{} & \swap{1} &\qw &\qw &\qw\rstick{3} \\
    \lstick{3} &\qw & \targX{} & \swap{1} &\qw &\qw\rstick{4} \\
    \lstick{\vdots} &\qw &\qw & \targX{} & \swap{1} &\qw\rstick{\vdots}  \\
    \lstick{$s$} &\qw &\qw &\qw & \targX{} &\qw\rstick{1}  \\
    \end{quantikz}.
    \label{D1}
\end{equation}

However, this generalization does not act as an effective filtering operation, as briefly mentioned in~\cite{schiffer2025quantum}. 
Instead, we consider generalizing the swap test circuit to a cyclic permutation test~\cite{kada2008efficiency,buhrman2024permutation,liu2025generalized}, depicted in Fig.~\ref{three-dev}(a). 
To apply the filter across $s$ ($s\geq 2$) devices, the top auxiliary qudit must have $s$ levels. The operator $F$ performs the discrete Fourier transform, defined as: $\ket{z}\rightarrow\frac{1}{\sqrt{s}}\sum_{q=0}^{s-1}\omega^{zq}\ket{q}$ where $\omega=e^{2\pi i/s}$. The multi-level controlled-$D$ operator $\sum_{q=0}^{s-1}\ket{q}\bra{q}\otimes D^q$ is then applied as follows:
\begin{equation}
    \begin{quantikz}[row sep = {0.5cm, between origins}]
    \lstick{$\sum_{q=0}^{s-1}a_q\ket{q}$}& \ctrl{1} &\qw \\
    \lstick{$\ket{\Psi}$}&\gate{D} &\qw \\
    \end{quantikz} \ = \ \sum_{q=0}^{s-1} a_q \ket{q} \otimes D^q \ket{\Psi},
    \label{Dops}
\end{equation}
where $a_q\in\mathbb{C}$ and $\sum_{q=0}^{s-1}|a_q|^2=1$. An inverse discrete Fourier transform $F^{\dagger}$ is then applied before measuring the top qudit. Similar to the two-device filter, we extend the postselection criteria correspondingly, and one can see that \emph{strong} selection can still make the states of each party remain product.
In addition, \emph{weak} postselection makes the state in general entangled but it evolves asymptotically to a product of identical eigenstates when $K$ is large. 
The pseudocode for the $s$-device filter is shown in Algorithm~\ref{algo-s}.

%\begin{figure}[t]
\begin{algorithm}[H]
\caption{$s$-device distributed filtering algorithm}
\label{algo-s}
    \begin{algorithmic}[1]
        \REQUIRE $s$ product states $\ket{\psi_1}\otimes\cdots\otimes\ket{\psi_s}=\ket{\psi^{(0)}}^{\otimes s}$; target iteration $K$; postselection type $\mathcal{P}$.
        \ENSURE State $\ket{\Psi^{(K)}_{12\cdots s}}$.
        \STATE Initialize: $\ket{\Psi^{(0)}} \gets \ket{\psi^{(0)}}^{\otimes s}$;
        \FOR{$k = 1$ to $K$}
            \STATE Generate random time $t_k$ such that $\varphi_j^{(k)} = -t_k \lambda_j \mod 2\pi \sim \text{Uniform}(0, 2\pi)$ for every $j$;
            \STATE Input $\ket{\Psi^{(k-1)}}$ and $t_k$ to the circuit in Fig.~\ref{three-dev}(a);
            \STATE Measure and get the result $\ket{m_k} = \ket{m^{(k)}_0 m^{(k)}_1 \cdots m^{(k)}_s}$;
            \IF{$\mathcal{P} = \text{\emph{weak}} \textbf{ and } m^{(k)}_0 = 0$}
                \STATE Output $\ket{\Psi^{(k)}}$ and \textbf{continue};
            \ELSIF{$\mathcal{P} = \text{\emph{strong}} \textbf{ and } (m_k = 00\cdots0 \textbf{ or } m_k = 01\cdots1)$}
                \STATE Output $\ket{\Psi^{(k)}}$ and \textbf{continue};
            \ELSE
                \STATE \textbf{break} and start over from initialization;
            \ENDIF
        \ENDFOR
        \RETURN $\ket{\Psi_{12\cdots s}^{(K)}}$.
    \end{algorithmic}
\end{algorithm}
%\end{figure}

\begin{figure*}
    \centering
    \includegraphics[width=1.0\linewidth]{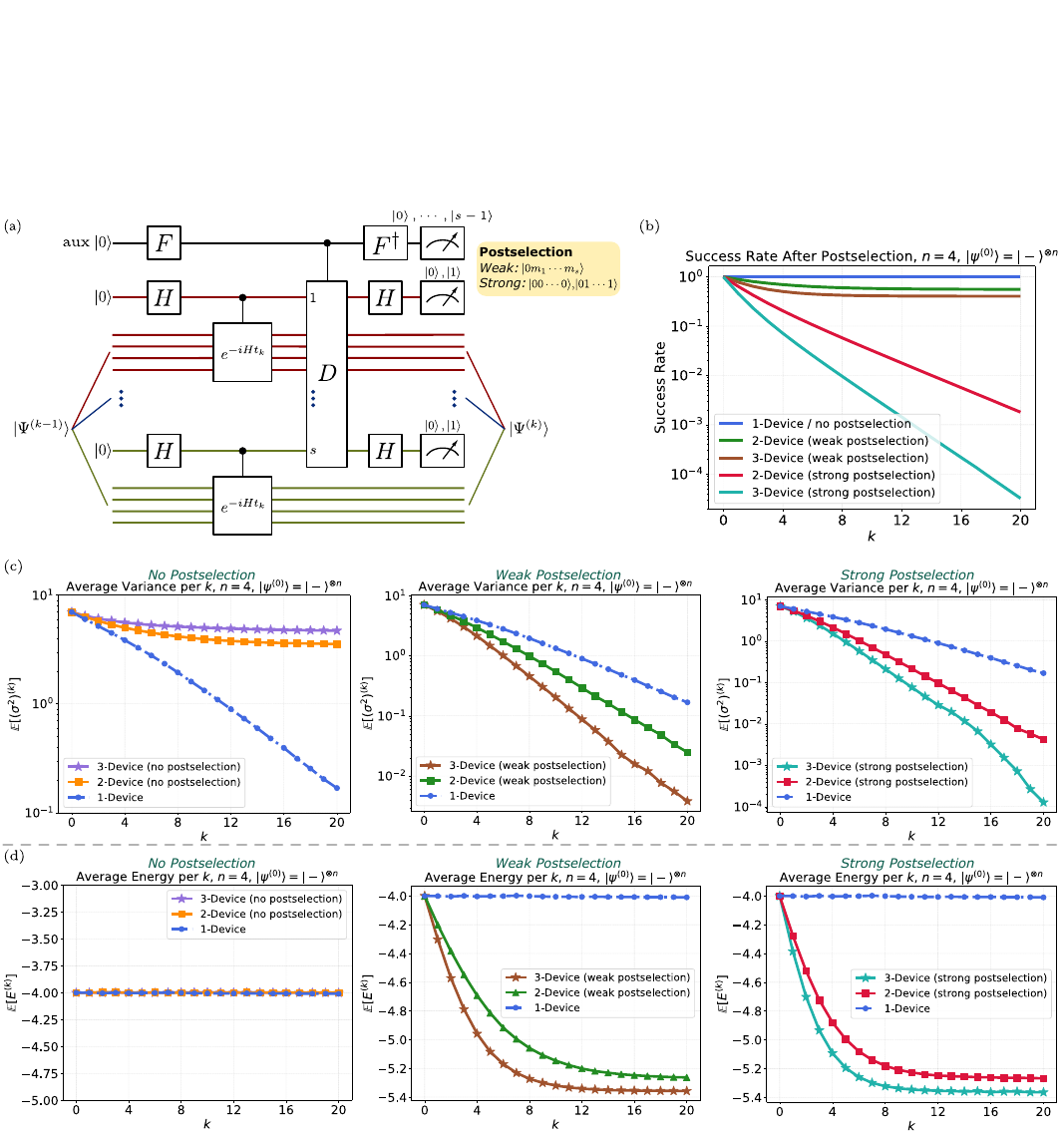}
    \caption{(a) Proposed distributed filtering algorithm across $s \geq 2$ devices. Similar to the two-device algorithm, the input is a product state $\ket{\Psi^{(0)}} = \ket{\psi^{(0)}}^{\otimes s}$ with sufficiently large random values $t_k$ for each iteration $k$. Additionally, a top auxiliary qudit with $s$ levels is introduced, and a cyclic permutation test~\cite{kada2008efficiency,buhrman2024permutation,liu2025generalized} is applied across the top auxiliary qudit and the Hadamard test auxiliary qubits of each device. Postselection, as defined above, is still required. Numerics for $n=4$ and $\ket{\psi^{(0)}}=\ket{-}^{\otimes n}$ for single-, two-, and three-device cases are shown in (b), (c), and (d): (b) cumulative success rate after postselection; (c) average variance; and (d) average energy of the pure state ensemble per iteration $k$ under no postselection, \emph{weak} postselection, and \emph{strong} postselection, respectively. We use $10^7$ trials from the beginning for \emph{strong} postselection cases and $10^5$ for the rest to acquire accurage enough expectation values.
    Plots with complete error bars are presented in Appendix~\ref{appx-complete}.}
    \label{three-dev}
\end{figure*}
% fix the y-label

We numerically benchmark the performance of a distributed filter algorithm on a single device, two or three devices in Fig.~\ref{three-dev}(b,c,d). We choose $n=4$ and $\ket{\psi^{(0)}}=\ket{+}^{\otimes n}$ per device. Notably, the three-device filter (with either \emph{weak} or \emph{strong} postselection) consistently yields states with lower average energy variance. It also introduces greater energy biases. Regarding the cumulative success rate, similar to the two-device case, \emph{strong} postselection incurs exponential overhead of the success rate with respect to the number of iterations $k$, while the overhead for \emph{weak} postselection stabilizes after several iterations, which is lower-bounded by $\sum_j|c_j^{(0)}|^{2s}$ for general $s$.
Additionally, the three-device algorithm always incurs more overhead, as the number of possible measurement results that are discarded under postselection increases.
Note that when the number of devices $s$ is prime, each $D^q$ for $q\neq0$ constitutes a full-cycle permutation on the registers~\cite{liu2025generalized}. However, this does not hold for non-prime $s$, which might affect the performance of the multi-device distributed filtering algorithm under \emph{weak} postselection.
We provide further analysis on this generalization in Appendix~\ref{appendix:3-dev}.

Numerical simulations for more than three devices are not feasible due to the exponential growth of Hilbert space. However, already from the data available for up to three devices and four qubits we can attempt an extrapolation of the protocol performance.
We perform an exponential fit $\propto\exp(-\eta k)$ of the average variance shown in Fig.~\ref{three-dev}(c) in the interval $k\in[4,12]$, i.e.,~after the initial roll-off and before the data for the \emph{strong} postselection case has some visible fluctuations due to a smaller number of samples remaining after postselection.
For the single device case, we find $\eta = 0.178$.
Using the \emph{weak} postselection criterion, we have $\eta = 0.277$ for two devices and $\eta = 0.386$ for three devices. 
For the \emph{strong} postselection criterion, we find  $\eta = 0.376$ for two devices and $\eta=0.490$ for three devices. 

With only limited data available due to the prohibitive simulation cost with the system size and the number of devices, an extrapolation to more than two devices is only indicative.
However, for the \emph{strong} postselection case there is a related analytical result in the limit of many devices~\cite{schiffer2025quantum}: the relative suppression between a dominant and a subdominant amplitude converges asymptotically to a fixed value. 
However, the data in Fig.~\ref{three-dev}(c) does not consider the ratio between amplitudes, but takes all amplitudes of the state into consideration. 
While an immediate quantitative comparison with the analytical result is therefore not possible, we conjecture that the scaling of the average variance of the \emph{strong} postselection result could also converge, implying that additional devices result in diminishing returns.

\section{Discussion and Outlook\label{sec4}}

In this work, we have presented a distributed filtering algorithm for eigenstate preparation that employs multiple quantum devices with postselection. By coordinating operations across two devices, this approach allows one to prepare low-variance states with a lower circuit depth than in a single-device setup. 
Our results demonstrate that postselection can play a critical role in accelerating convergence towards the eigenstates. We emphasize that only the overhead for the \emph{strong} postselection grows exponentially in the number of rounds. For \emph{weak} postselection, the overhead due to postselection converges to a fixed value that is independent of the number of rounds of the protocol.

The protocol we present is a distributed algorithm inspired by methods from iterative quantum phase estimation (IQPE). 
Our technique inherits the minimal requirement of a single auxiliary qubit for the conditional dynamics from IQPE.
Textbook QPE, however, proceeds by including an inverse quantum Fourier transformation on the auxiliary register. 
One can devise a similar distributed algorithm for textbook QPE, although this would require additional teleportation operations per circuit between the two devices due to the additional auxiliary qubits than the ones considered in this work.

The two-device circuit was previously presented in~\cite{schiffer2025quantum} with an eigenstate broadcasting application, so we briefly highlight the novel contributions in this work. First, we no longer assume the initial state to be in a superposition of a dominant eigenstate and other subdominant state. Rather, we consider an energy distribution described by its mean and variance, corresponding to the typical scenario of an initial product states~\cite{rai2024matrix}.
Second, we introduced novel technical analysis tools based on the moments of the distribution. 
Next, we have included a new numerical analysis on the different use (which we named \emph{weak} vs.~\emph{strong}) of postselection in the distributed quantum algorithm. 
Finally, we present a new approach to realizing a symmetric projection of the auxiliary qubits that uses a qudit-controlled derangement operator.

The extension of the distributed algorithms to arbitrarily many connected quantum devices is a main result of our paper. We also show numerically that a third device allows for an even larger reduction of the energy variance for a fixed number of iterations of the protocol than with two devices. This naturally leads to an increased reset overhead, which is also bounded for the \emph{weak} postselection criterion, and increases exponentially with each iteration for the \emph{strong} postselection criterion. However, when successful, \emph{strong} postselection converges exponentially faster than the corresponding \emph{weak} postselection scheme.

Our findings also suggest that a hybrid protocol that starts with \emph{strong} postselection and transitions to \emph{weak} postselection would optimize resources better, as the \emph{weak} postselection variance plateaus when the devices hold a sufficiently good approximation to identical eigenstates.

Overall, distributed quantum algorithms offer promising strategies to accelerate eigenstate preparation by reducing the circuit depth required. 
Practical realizations of these algorithms in noisy quantum devices may shed light on optimal engineering of projections onto symmetric subspaces and further improvements of distributed techniques more generally.
\vspace{10mm}

\begin{acknowledgments}
The numerical experiments of this work were performed using the compute resources from the Academic Leiden Interdisciplinary Cluster Environment (ALICE) provided by Leiden University.
X.L.~thanks Kshiti Sneh Rai, Johannes Knörzer, Zherui Jerry Wang and Yuning Zhang for the helpful discussions.
B.F.S.~acknowledges funding from the the Federal Ministry of Education and Research Germany (BMBF) via the project FermiQP (13N15889).
Work in Munich was part of the Munich Quantum Valley, which is supported by the Bavarian state government with funds from the Hightech Agenda Bayern Plus.
J.T.~acknowledges the support received from the European Union’s Horizon Europe research and innovation programme through the ERC StG FINETEA-SQUAD (Grant No. 101040729). J.T.~also acknowledges the support received by the Dutch National Growth Fund (NGF), as part of the Quantum Delta NL programme. This publication is part of the ‘Quantum Inspire – the Dutch Quantum Computer in the Cloud’ project (with project number [NWA.1292.19.194]) of the NWA research program ‘Research on Routes by Consortia (ORC)’, which is funded by the Netherlands Organization for Scientific Research (NWO). 
The views and opinions expressed here are solely those of the authors and do not necessarily reflect those of the funding institutions. Neither of the funding institutions can be held responsible for them.
\end{acknowledgments}

\bibliographystyle{quantum}
\bibliography{Dist_Filter}

\begin{thebibliography}{10}

\bibitem{nielsen2010quantum}
Michael~A. Nielsen and Isaac~L. Chuang.
\newblock ``Quantum {{Computation}} and {{Quantum Information}}: 10th {{Anniversary Edition}}''.
\newblock Cambridge University Press. ~(2010).
\newblock  url:~\url{https://doi.org/10.1017/CBO9780511976667}.

\bibitem{kitaev1995quantum}
A.~Yu Kitaev.
\newblock ``Quantum measurements and the {{Abelian Stabilizer Problem}}''~(1995).
\newblock  \href{http://arxiv.org/abs/quant-ph/9511026}{arXiv:quant-ph/9511026}.

\bibitem{obrien2019quantum}
Thomas~E. O'Brien, Brian Tarasinski, and Barbara~M. Terhal.
\newblock ``Quantum phase estimation of multiple eigenvalues for small-scale (noisy) experiments''.
\newblock \href{https://dx.doi.org/10.1088/1367-2630/aafb8e}{New Journal of Physics {\bf 21}, 023022}~(2019).

\bibitem{dutkiewicz2022heisenberglimited}
Alicja Dutkiewicz, Barbara~M. Terhal, and Thomas~E. O'Brien.
\newblock ``Heisenberg-limited quantum phase estimation of multiple eigenvalues with few control qubits''.
\newblock \href{https://dx.doi.org/10.22331/q-2022-10-06-830}{Quantum {\bf 6}, 830}~(2022).

\bibitem{hartmann2005spectral}
Michael Hartmann, G{\"u}nter Mahler, and Ortwin Hess.
\newblock ``Spectral {{Densities}} and {{Partition Functions}} of {{Modular Quantum Systems}} as {{Derived}} from a {{Central Limit Theorem}}''.
\newblock \href{https://dx.doi.org/10.1007/s10955-004-4298-5}{Journal of Statistical Physics {\bf 119}, 1139--1151}~(2005).

\bibitem{abrams1999quantum}
Daniel~S. Abrams and Seth Lloyd.
\newblock ``Quantum {{Algorithm Providing Exponential Speed Increase}} for {{Finding Eigenvalues}} and {{Eigenvectors}}''.
\newblock \href{https://dx.doi.org/10.1103/PhysRevLett.83.5162}{Physical Review Letters {\bf 83}, 5162--5165}~(1999).

\bibitem{poulin2009preparing}
David Poulin and Pawel Wocjan.
\newblock ``Preparing {{Ground States}} of {{Quantum Many-Body Systems}} on a {{Quantum Computer}}''.
\newblock \href{https://dx.doi.org/10.1103/PhysRevLett.102.130503}{Physical Review Letters {\bf 102}, 130503}~(2009).

\bibitem{ge2019faster}
Yimin Ge, Jordi Tura, and J.~Ignacio Cirac.
\newblock ``Faster ground state preparation and high-precision ground energy estimation with fewer qubits''.
\newblock \href{https://dx.doi.org/10.1063/1.5027484}{Journal of Mathematical Physics {\bf 60}, 022202}~(2019).

\bibitem{lu2021algorithms}
Sirui Lu, Mari~Carmen Ba{\~n}uls, and J.~Ignacio Cirac.
\newblock ``Algorithms for {{Quantum Simulation}} at {{Finite Energies}}''.
\newblock \href{https://dx.doi.org/10.1103/PRXQuantum.2.020321}{PRX Quantum {\bf 2}, 020321}~(2021).

\bibitem{irmejs2024efficient}
Reinis Irmejs, Mari~Carmen Ba{\~n}uls, and J.~Ignacio Cirac.
\newblock ``Efficient {{Quantum Algorithm}} for {{Filtering Product States}}''.
\newblock \href{https://dx.doi.org/10.22331/q-2024-06-27-1389}{Quantum {\bf 8}, 1389}~(2024).

\bibitem{rai2024matrix}
Kshiti~Sneh Rai, J.~Ignacio Cirac, and {\'A}lvaro~M. Alhambra.
\newblock ``Matrix product state approximations to quantum states of low energy variance''.
\newblock \href{https://dx.doi.org/10.22331/q-2024-07-10-1401}{Quantum {\bf 8}, 1401}~(2024).

\bibitem{dobsicek2007arbitrary}
Miroslav Dob{\v s}{\'i}{\v c}ek, G{\"o}ran Johansson, Vitaly Shumeiko, and G{\"o}ran Wendin.
\newblock ``Arbitrary accuracy iterative quantum phase estimation algorithm using a single ancillary qubit: {{A}} two-qubit benchmark''.
\newblock \href{https://dx.doi.org/10.1103/PhysRevA.76.030306}{Physical Review A {\bf 76}, 030306}~(2007).

\bibitem{xu2014demonlike}
Jin-Shi Xu, Man-Hong Yung, Xiao-Ye Xu, Sergio Boixo, Zheng-Wei Zhou, Chuan-Feng Li, Al{\'a}n {Aspuru-Guzik}, and Guang-Can Guo.
\newblock ``Demon-like algorithmic quantum cooling and its realization with quantum optics''.
\newblock \href{https://dx.doi.org/10.1038/nphoton.2013.354}{Nature Photonics {\bf 8}, 113--118}~(2014).

\bibitem{meister2022resourcefrugal}
Richard Meister and Simon~C. Benjamin.
\newblock ``Resource-frugal {{Hamiltonian}} eigenstate preparation via repeated quantum phase estimation measurements''~(2022).
\newblock  \href{http://arxiv.org/abs/2212.00846}{arXiv:2212.00846}.

\bibitem{chen2020quantum}
Yanzhu Chen and Tzu-Chieh Wei.
\newblock ``Quantum algorithm for spectral projection by measuring an ancilla iteratively''.
\newblock \href{https://dx.doi.org/10.1103/PhysRevA.101.032339}{Physical Review A {\bf 101}, 032339}~(2020).

\bibitem{schiffer2025quantum}
Benjamin~F. Schiffer and Jordi Tura.
\newblock ``Quantum eigenstate preparation assisted by a coherent link''.
\newblock \href{https://dx.doi.org/10.1103/PhysRevA.111.012445}{Physical Review A {\bf 111}, 012445}~(2025).

\bibitem{qian2024demonstration}
Zhengrong Qian, Jacob Watkins, Gabriel Given, Joey Bonitati, Kenneth Choi, and Dean Lee.
\newblock ``Demonstration of the rodeo algorithm on a quantum computer''.
\newblock \href{https://dx.doi.org/10.1140/epja/s10050-024-01373-9}{The European Physical Journal A {\bf 60}, 151}~(2024).

\bibitem{choi2021rodeo}
Kenneth Choi, Dean Lee, Joey Bonitati, Zhengrong Qian, and Jacob Watkins.
\newblock ``Rodeo {{Algorithm}} for {{Quantum Computing}}''.
\newblock \href{https://dx.doi.org/10.1103/PhysRevLett.127.040505}{Physical Review Letters {\bf 127}, 040505}~(2021).

\bibitem{aharonov2006polynomial}
Dorit Aharonov, Vaughan Jones, and Zeph Landau.
\newblock ``A polynomial quantum algorithm for approximating the {{Jones}} polynomial''.
\newblock In Proceedings of the Thirty-Eighth Annual {{ACM}} Symposium on {{Theory}} of {{Computing}}.
\newblock \href{https://dx.doi.org/10.1145/1132516.1132579}{Pages 427--436}.
\newblock {{STOC}} '06New York, NY, USA~(2006). Association for Computing Machinery.

\bibitem{magnard2020microwave}
P.~Magnard, S.~Storz, P.~Kurpiers, J.~Sch{\"a}r, F.~Marxer, J.~L{\"u}tolf, T.~Walter, J.-C. Besse, M.~Gabureac, K.~Reuer, A.~Akin, B.~Royer, A.~Blais, and A.~Wallraff.
\newblock ``Microwave {{Quantum Link}} between {{Superconducting Circuits Housed}} in {{Spatially Separated Cryogenic Systems}}''.
\newblock \href{https://dx.doi.org/10.1103/PhysRevLett.125.260502}{Physical Review Letters {\bf 125}, 260502}~(2020).

\bibitem{harvey-collard2022coherent}
Patrick {Harvey-Collard}, Jurgen Dijkema, Guoji Zheng, Amir Sammak, Giordano Scappucci, and Lieven M.~K. Vandersypen.
\newblock ``Coherent {{Spin-Spin Coupling Mediated}} by {{Virtual Microwave Photons}}''.
\newblock \href{https://dx.doi.org/10.1103/PhysRevX.12.021026}{Physical Review X {\bf 12}, 021026}~(2022).

\bibitem{cirac1999distributed}
J.~I. Cirac, A.~K. Ekert, S.~F. Huelga, and C.~Macchiavello.
\newblock ``Distributed quantum computation over noisy channels''.
\newblock \href{https://dx.doi.org/10.1103/PhysRevA.59.4249}{Physical Review A {\bf 59}, 4249--4254}~(1999).

\bibitem{buhrman2003distributed}
Harry Buhrman and Hein R{\"o}hrig.
\newblock ``Distributed {{Quantum Computing}}''.
\newblock In Branislav Rovan and Peter Vojt{\'a}{\v s}, editors, Mathematical {{Foundations}} of {{Computer Science}} 2003.
\newblock \href{https://dx.doi.org/10.1007/978-3-540-45138-9_1}{Pages 1--20}.
\newblock Berlin, Heidelberg~(2003). Springer.

\bibitem{barenco1997stabilization}
Adriano Barenco, Andr{\'e} Berthiaume, David Deutsch, Artur Ekert, Richard Jozsa, and Chiara Macchiavello.
\newblock ``Stabilization of {{Quantum Computations}} by {{Symmetrization}}''.
\newblock \href{https://dx.doi.org/10.1137/S0097539796302452}{SIAM Journal on Computing {\bf 26}, 1541--1557}~(1997).

\bibitem{buhrman2001quantum}
Harry Buhrman, Richard Cleve, John Watrous, and Ronald {de Wolf}.
\newblock ``Quantum {{Fingerprinting}}''.
\newblock \href{https://dx.doi.org/10.1103/PhysRevLett.87.167902}{Physical Review Letters {\bf 87}, 167902}~(2001).

\bibitem{born1926zur}
Max Born.
\newblock ``{Zur Quantenmechanik der Sto{\ss}vorg{\"a}nge}''.
\newblock \href{https://dx.doi.org/10.1007/BF01397477}{Zeitschrift f{\"u}r Physik {\bf 37}, 863--867}~(1926).

\bibitem{kada2008efficiency}
Masaru Kada, Harumichi Nishimura, and Tomoyuki Yamakami.
\newblock ``The efficiency of quantum identity testing of multiple states''.
\newblock \href{https://dx.doi.org/10.1088/1751-8113/41/39/395309}{Journal of Physics A: Mathematical and Theoretical {\bf 41}, 395309}~(2008).

\bibitem{buhrman2024permutation}
Harry Buhrman, Dmitry Grinko, Philip~Verduyn Lunel, and Jordi Weggemans.
\newblock ``Permutation tests for quantum state identity''~(2024).
\newblock  \href{http://arxiv.org/abs/2405.09626}{arXiv:2405.09626}.

\bibitem{liu2025generalized}
Xiaoyu Liu, Johannes Kn{\"o}rzer, Zherui~Jerry Wang, and Jordi Tura.
\newblock ``Generalized concentratable entanglement via parallelized permutation tests''.
\newblock \href{https://dx.doi.org/10.1103/jtlj-qs3y}{Physical Review Research {\bf 7}, L032022}~(2025).

\bibitem{weyl1916ueber}
Hermann Weyl.
\newblock ``{{\"U}ber die Gleichverteilung von Zahlen mod. Eins}''.
\newblock \href{https://dx.doi.org/10.1007/BF01475864}{Mathematische Annalen {\bf 77}, 313--352}~(1916).

\end{thebibliography}

\onecolumngrid

\clearpage
\appendix

\section{Single-device filter}\label{appx-single}

We start from stating the single-device filter algorithm~\cite{abrams1999quantum, xu2014demonlike, chen2020quantum, meister2022resourcefrugal, schiffer2025quantum} and the pseudocode is shown below:

\begin{algorithm}[H]
\caption{Single-device distributed ancilla-mediated filtering algorithms.}
\label{algo3}
\begin{algorithmic}[1]
\REQUIRE A product state $\ket{\psi^{(0)}}$; target iteration $K$.
\ENSURE State $\ket{\psi^{(K)}}$.
\STATE Initialize: $\ket{\Psi^{(0)}} \gets \ket{\psi^{(0)}}$;
\FOR{$k = 1$ to $K$}
    \STATE Generate random time $t_k$ such that $\varphi_j^{(k)} = -t_k \lambda_j \mod 2\pi \sim \text{Uniform}(0, 2\pi)$ for every $j$;
    \STATE Input $\ket{\Psi^{(k-1)}}$ and $t_k$ to the circuit in Fig.~\ref{single-dev};
    \STATE Measure and output $\ket{\Psi^{(k)}}$;
    \STATE \textbf{continue};
\ENDFOR
\STATE $\ket{\psi^{(K)}} \gets \ket{\Psi^{(K)}}$;
\RETURN $\ket{\psi^{(K)}}$.
\end{algorithmic}
\end{algorithm}

Similarly, we would like to find out the behavior of the output states' average $E$ and $\sigma^2$.

\begin{figure}[h]
    \centering
    \includegraphics[width=0.5\linewidth]{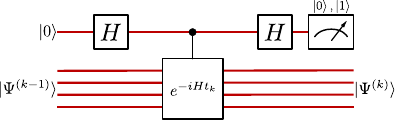}
    \caption{Circuit for single-device filtering algorithm.}
    \label{single-dev}
\end{figure}

\subsection{Energy}

After $K\in \mathbb{Z}^{+}$ iterations, the output state $\ket{\psi^{(K)}}$ has the energy:
\begin{equation}
    E^{(K)} = \bra{\psi^{(K)}}H\ket{\psi^{(K)}}.
\end{equation}
As mentioned in the main text, for each repetition, $\ket{\psi^{(K)}}$ can vary. Suppose we have enough repetitions $r$ which corresponds to each state $\ket{\psi_r^{(K)}}$ with probability $p_r$. Therefore, assume the sequence of time evolution $\mathbf{t}=\{t_1,t_2,\cdots,t_K\}$ on average, we have:
\begin{equation}
    \mathbb{E} \left[ E^{(K)} \Big| \mathbf{t} \right] = \sum_{r}p_r\bra{\psi_r^{(K)}} H \ket{\psi_r^{(K)}}=\Tr \left(H \sum_r p_r\ket{\psi_r^{(K)}}\bra{\psi_r^{(K)}}\right) = \Tr\left(H\rho_{\mathbf{t}}^{(K)}\right),
\end{equation}
where we denote $\rho_{\mathbf{t}}^{(K)}$ as the density matrix at iteration $K$ under the fixed time evolution sequence $\mathbf{t}$. Note that for the Hadamard test operation at each iteration $k$, we have a channel:
\begin{equation}
\rho_{\mathbf{t}}^{(k)}=\mathcal{E}^{(k)}(\rho^{(k-1)})=\sum_{m=0}^{1}\mathbf{E}_m^{(k)}\rho^{(k-1)}{\mathbf{E}_m^{(k)}}^{\dagger},
\end{equation}
with Kraus operators:
\begin{equation}
\begin{split}
    \mathbf{E}^{(k)}_0 &= \frac{1}{2}(I+U(t_k)) \ \ \ \ \ \text{$\rightarrow$ $\ket{0}$ on the auxiliary qubit}; \\
    \mathbf{E}^{(k)}_1 &= \frac{1}{2}(I-U(t_k)) \ \ \ \ \ \text{$\rightarrow$ $\ket{1}$ on the auxiliary qubit}.
\end{split}
\end{equation}
Therefore,
\begin{equation}
\begin{split}
    \rho_{\mathbf{t}}^{(K)}  = \mathcal{E}^{(K)} \circ \mathcal{E}^{(K-1)} \circ \cdots \circ \mathcal{E}^{(1)}\left(\ket{\psi^{(0)}}\bra{\psi^{(0)}}\right). 
\end{split}
\end{equation}

Note that:
\begin{equation}
    \rho_{\mathbf{t}}^{(k)} = \mathcal{E}^{(k)} \left(\rho^{(k-1)}\right) = \frac{1}{2} \left( \rho^{(k-1)} + U(t_k) \rho^{(k-1)} U(t_k)^{\dagger} \right),
\end{equation}
including the input state, controlled unitaries and Hamiltonian into the eigenbasis, i.e.,
\begin{equation}
    \ket{\psi^{(0)}}=\sum_{j}c_j^{(0)}\ket{\phi_j};
\end{equation}
\begin{equation}
    U(t_k)=e^{-it_k H}=\sum_{j}e^{-it_k \lambda_j}\ket{\phi_j}\bra{\phi_j} = \sum_{j}e^{i\varphi^{(k)}_j}\ket{\phi_j}\bra{\phi_j};
\end{equation}
\begin{equation}
    H = \sum_{j}\lambda_j \ket{\phi_j}\bra{\phi_j},
\end{equation}
we have (we omit superscript $^{(0)}$ of $c_j^{(0)}$ for simplicity):
\begin{equation}
\begin{split}
    \rho_{\mathbf{t}}^{(K)}  &= \mathcal{E}^{(K)} \circ \mathcal{E}^{(K-1)} \circ \cdots \circ \mathcal{E}^{(1)}\left(\ket{\psi^{(0)}}\bra{\psi^{(0)}}\right)  \\
                &= \frac{1}{2^K}\sum_{\alpha \in \mathcal{P}(\{1,2,\cdots,K\})} \left(\prod_{a\in\alpha}U(t_a)\right) \ket{\psi^{(0)}} \bra{\psi^{(0)}} \left(\prod_{a'\in\alpha}U(t_{a'})\right)^{\dagger} \\
                &= \frac{1}{2^K}\sum_{\alpha \in \mathcal{P}(\{1,2,\cdots,K\})} \sum_{l,l'} \exp\left(i\sum_{a\in\alpha}\varphi_l^{(a)}-i\sum_{a'\in\alpha}\varphi_{l'}^{(a')}\right)c_lc_{l'}^{*}\ket{\phi_l}\bra{\phi_{l'}},
\end{split}
\end{equation}
where $\mathcal{P}(\{1,2,\cdots,K\})$ denotes all the subsets of $\{1,2,\cdots,K\}$. Therefore:
\begin{equation}
    \begin{split}
        \mathbb{E}\left[E^{(K)} \Big| \mathbf{t} \right]  &= \Tr\left(H\rho_{\mathbf{t}}^{(K)}\right) \\
                            &= \frac{1}{2^K}\sum_{\alpha \in \mathcal{P}(\{1,2,\cdots,K\})}\Tr\left(\sum_{j}\lambda_j \ket{\phi_j}\bra{\phi_j}   \sum_{l,l'} \exp\left(i\sum_{a\in\alpha}\varphi_l^{(a)}-i\sum_{a'\in\alpha}\varphi_{l'}^{(a')}\right)c_lc_{l'}^{*}\ket{\phi_l}\bra{\phi_{l'}} \right) \\
                            &= \Tr\left(\sum_{l,l'} \lambda_l c_lc_{l'}^{*}\ket{\phi_l}\bra{\phi_{l'}}\right) = \sum_{l}\lambda_l \left| c_l \right|^2.
    \end{split}
\end{equation}
As $\mathbb{E}\left[E^{(K)} \Big| \mathbf{t} \right]$ does not depend on $\mathbf{t}$, therefore:
\begin{equation}
    \mathbb{E}\left[ E^{(k)} \right] = \sum_{j} \lambda_j \left| c_j \right|^2. 
\end{equation}
which equals to the energy of the initial state $\ket{\psi^{(0)}}$.

\subsection{Energy Variance} \label{appx-single-var}

The variance of the output state $\ket{\psi^{(K)}}$ is:
\begin{equation}
    \left(\sigma^2\right)^{(K)}=\bra{\psi^{(K)}} H^2 \ket{\psi^{(K)}} - \left(\bra{\psi^{(K)}}H\ket{\psi^{(K)}}\right)^2,
\end{equation}
and on average:
\begin{equation}
    \mathbb{E}\left[ \left(\sigma^{2}\right)^{(K)}  \right] = \mathbb{E}\left[\bra{\psi^{(K)}} H^2 \ket{\psi^{(K)}}\right] - \mathbb{E}\left[\left(\bra{\psi^{(K)}}H\ket{\psi^{(K)}}\right)^2\right].
\end{equation}
Similar to the previous derivations, we have the first term as:
\begin{equation}
    \mathbb{E}\left[\bra{\psi^{(K)}} H^2 \ket{\psi^{(K)}}\right] = \Tr\left(H^2\rho^{(K)}\right) = \sum_{j}\lambda_j^2\left| c_j \right|^2.
\end{equation}
The second term is complicated. First, we again fix the time sequence $\mathbf{t}$ and find:
\begin{equation}
    \mathbb{E}\left[\left( \bra{\psi^{(K)}} H \ket{\psi^{(K)}} \right)^2 \Big| \mathbf{t} \right] = \sum_r p_r \left( \bra{\psi_r^{(K)}} H \ket{\psi_r^{(K)}} \right)^2.
\end{equation}
Note that each $\ket{\psi^{(K)}_r}$ can be denoted as:
\begin{equation}
    \ket{\psi_r^{(K)}} = \frac{1}{\sqrt{P_{\mathbf{m}}}}\frac{1}{2^K}\prod_{k=1}^{K}(I+(-1)^{m_k}U(t_k))\ket{\psi^{(0)}}=\frac{1}{\sqrt{P_{\mathbf{m}}}}\frac{1}{2^K}\sum_{j} c_j \prod_{k=1}^{K}\left(1+(-1)^{m_k}e^{i\varphi_j^{(k)}}\right)\ket{\phi_j},
    \label{psirK}
\end{equation}
where $P_{\mathbf{m}}$ denotes the probability of getting the measurement result $\mathbf{m}=\{m_1,m_2,\cdots, m_K\}$ on the auxiliary qubit of each iteration $k$ at repetition $r$. Then we have:
\begin{equation}
    \bra{\psi_r^{(K)}} H \ket{\psi_r^{(K)}}=\frac{1}{P_{\mathbf{m}}} \frac{1}{2^K} \sum_{j} \left|c_j\right|^2 \prod_{k=1}^{K}\left(1+(-1)^{m_k}\cos \varphi^{(k)}_{j}\right)\lambda_j,
\end{equation}
and:
\begin{equation}
    P_{\mathbf{m}}= \frac{1}{2^K}\sum_{j} \left|c_j\right|^2 \prod_{k=1}^{K} \left(1+(-1)^{m_k}\cos \varphi^{(k)}_{j}\right).
    \label{Pm}
\end{equation}
Therefore:
\begin{equation}
\begin{split}
    \mathbb{E}\left[\left( \bra{\psi^{(K)}} H \ket{\psi^{(K)}} \right)^2 \Big| \mathbf{t} \right] &= \sum_r p_r \left( \bra{\psi_r^{(K)}} H \ket{\psi_r^{(K)}} \right)^2 \\
    &= \frac{1}{2^K}\sum_{\mathbf{m}} \frac{\left(\sum_{j} \left|c_j\right|^2\lambda_j \prod_{k=1}^{K}\left(1+(-1)^{m_k}\cos \varphi^{(k)}_{j}\right)\right)^2}{\sum_{j} \left|c_j\right|^2 \prod_{k=1}^{K} \left(1+(-1)^{m_k}\cos \varphi^{(k)}_{j}\right)} \\
    &= \frac{1}{2^K}\sum_{\mathbf{m}} \frac{\left(\sum_{j} \left|c_j\right|^2\lambda_j \prod_{k=1}^{K}\left(1+(-1)^{m_k}\cos \lambda_j t_k\right)\right)^2}{\sum_{j} \left|c_j\right|^2 \prod_{k=1}^{K} \left(1+(-1)^{m_k}\cos \lambda_j t_k\right)}.
\end{split}
\end{equation}
Now, suppose $\mathbf{t}$ are uniformly distributed within a time interval $T$. Then, we have:
\begin{equation}
\begin{split}
    \mathbb{E}\left[\left( \bra{\psi^{(K)}} H \ket{\psi^{(K)}} \right)^2 \right] 
    &= \frac{1}{T^K}\frac{1}{2^K}\sum_{\mathbf{m}} \int_0^T\cdots\int_0^T \frac{\left(\sum_{j} \left|c_j\right|^2\lambda_j \prod_{k=1}^{K}\left(1+(-1)^{m_k}\cos \lambda_j t_k\right)\right)^2}{\sum_{j} \left|c_j\right|^2 \prod_{k=1}^{K} \left(1+(-1)^{m_k}\cos \lambda_j t_k\right)}dt_1\cdots dt_K. 
\end{split}
\end{equation}
Next, we show why we need large $T$ to maximize the quantity $\mathbb{E}\left[\left( \bra{\psi^{(K)}} H \ket{\psi^{(K)}} \right)^2 \right]$, thus minimize $\mathbb{E}\left[ \left(\sigma^{2}\right)^{(K)}  \right]$. Since we lack prior knowledge of the populations of $\ket{\psi^{(0)}}$ in the eigenbasis, one of the optimal approaches is to completely randomize the phases $\varphi_j^{(k)}$. Specifically, from Eq.~\eqref{Pm} and~\eqref{psirK}, we can observe that: suppose one of the populations $\left| c^{(k-1)}_j \right|^2$ is already sufficiently larger than the other populations at $(k-1)$-th iteration. And we assume for the $k$-th iteration, the extra factor $ \left(1+(-1)^{m_k}\cos \varphi^{(k)}_{j}\right) > 1$. This, in general, magnifies the probability $P_{\mathbf{m}}$ of achieving $\left| c^{(k)}_j \right|^2 = \left| c^{(k-1)}_j \right|^2 \left(1+(-1)^{m_k}\cos \varphi^{(k)}_{j}\right) > \left| c^{(k-1)}_j \right|^2$, amplifying this population even more. In order to achieve this, we need to let $\varphi_j^{(k)}$ be independent with each other, which requires large time interval $T$. Otherwise, the correlations among different $\varphi_j^{(k)}$ may magnify several populations instead of only one, turning into superpositions of eigenstates.
\begin{proposition}
\label{prop1}
    Under the assumptions that $\frac{\lambda_{j'}}{\lambda_{j''}}$ is (almost) irrational for any $j'\neq j''\in\{1,2,\cdots,2^n\}$, for the time interval $T$ in large limit and $t_k$ uniformly sampled within the interval, $\varphi_j^{(k)}=-t_k\lambda_j \Mod{2\pi}$ can be approximated as i.i.d. in $\text{Uniform}(0,2\pi)$ for all $j$ and $k$. 
\end{proposition}
\begin{proof}
Firstly, we consider the Weyl's criterion of equidistribution~\cite{weyl1916ueber}, showing that $\varphi_j^{(k)}$ for specific $j$ and random $t_k$ distributed in large interval $T$ can be considered as $\text{Uniform}(0,2\pi)$. Weyl's criterion states that $\frac{\varphi_j^{(k)}}{2\pi} \Mod{1}$ is uniformly distributed if and only if for all non-zero integers $l$:
\begin{equation}
    \lim_{R\rightarrow\infty}\frac{1}{R}\sum_{r=1}^{R}e^{i2\pi l \frac{\varphi_j^{(k)}}{2\pi}} = \lim_{R\rightarrow\infty}\frac{1}{R}\sum_{r=1}^{R}e^{-i l t_{k,r}\lambda_j} = \frac{1}{T}\int_{T_{min}}^{T_{max}} e^{-i l t_{k}\lambda_j}dt_k  = 0.
\end{equation}
Here we suppose $t_k\sim\text{Uniform}(T_{min}, T_{max})$, $T_{max}>T_{min}>0$ and $T_{max}-T_{min}=T$. By integrating the above integral, we have:
\begin{equation}
    \left|\frac{1}{T}\int_{T_{min}}^{T_{max}} e^{-i l t_{k}\lambda_j}dt_k\right|^2 = \frac{1}{l^2}\text{sinc}^2\left( \frac{T\lambda_j}{2} \right) \leq \text{sinc}^2\left( \frac{T\lambda_j}{2} \right) \rightarrow 0,
\end{equation}
for large interval $T$. 

Secondly, we show that each $\varphi_j^{(k)}$ is independent with each other for all $j$ and $k$, under the assumptions that $\frac{\lambda_{j'}}{\lambda_{j''}}$ is irrational for any $j'\neq j''\in\{1,2,\cdots,2^n\}$. If $\frac{\lambda_{j'}}{\lambda_{j''}}$ is rational, i.e., $\frac{\lambda_{j'}}{\lambda_{j''}}=\frac{p}{q}$ where $p,q\in \mathbb{Z}$, then there exists a period $\mathrm{t}$ within $T_{min}\sim T_{max}$ such that $\frac{2\pi p}{\lambda_{j'}}=\frac{2\pi q}{\lambda_{j''}}=\mathrm{t}$. Because of this, the points with coordinates $\left(\varphi_{j'}^{(k)},\varphi_{j''}^{(k)}\right)$ cannot cover the whole area of $(0,2\pi)^2$. However, if $\frac{\lambda_{j'}}{\lambda_{j''}}$ is irrational, there will be no such $\mathrm{t}$ that exists. Thus, $\left(\varphi_{j'}^{(k)},\varphi_{j''}^{(k)}\right)$ can cover all $(0,2\pi)^2$ if $t_k$ is sufficiently large. The independence can also be found for rational $\frac{\lambda_{j'}}{\lambda_{j''}}$ but with extremely large period $\mathrm{t}$. In this case, $\left(\varphi_{j'}^{(k)},\varphi_{j''}^{(k)}\right)$ still nearly cover $(0,2\pi)^2$, before any periodic behavior emerges. This is where we name almost-irrational $\frac{\lambda_{j'}}{\lambda_{j''}}$.

A simple counter-example for the independence is that we choose the Hamiltonian as $H=\sum_{j}Z_jZ_{j+1}$ where it has symmetric eigenvalues like $\lambda_{j'}=-\lambda_{j''}$. In this case, $\varphi_{j'}^{(k)}$ and $\varphi_{j''}^{(k)}$ are strongly correlated. The points $\left(\varphi_{j'}^{(k)},\varphi_{j''}^{(k)}\right)$ form a diagonal line on $(0,2\pi)^2$. Consequently, this algorithm fails as it cannot result in $\left| c_{j'} \right|^2$ being magnified while $\left| c_{j''} \right|^2$ is suppressed, or vice versa. Shifting the Hamiltonian to break the symmetries can remedy this, i.e., $H'=H+\epsilon I$ (cf.~Appendix in \cite{schiffer2025quantum}). 
\end{proof}
Now, by proving Proposition~\ref{prop1}, we have:
\begin{equation}
\begin{split}
    &\mathbb{E}\left[\left( \bra{\psi^{(K)}} H \ket{\psi^{(K)}} \right)^2 \right] \\
    =&\frac{1}{2^K}\sum_{\mathbf{m}} \left(\frac{1}{2\pi} \right)^{NK} \int_{0}^{2\pi}\cdots\int_{0}^{2\pi}  \frac{\left(\sum_{j} \left|c_j\right|^2\lambda_j \prod_{k=1}^{K}\left(1+(-1)^{m_k}\cos \varphi^{(k)}_{j}\right)\right)^2}{\sum_{j} \left|c_j\right|^2 \prod_{k=1}^{K} \left(1+(-1)^{m_k}\cos \varphi^{(k)}_{j}\right)} d\varphi_{1}^{(1)}\cdots d\varphi_{N}^{(K)} \\
    =&\left(\frac{1}{2\pi} \right)^{NK} \int_{0}^{2\pi}\cdots\int_{0}^{2\pi}  \frac{\left(\sum_{j} \left|c_j\right|^2\lambda_j \prod_{k=1}^{K}\left(1-\cos \varphi^{(k)}_{j}\right)\right)^2}{\sum_{j} \left|c_j\right|^2 \prod_{k=1}^{K} \left(1-\cos \varphi^{(k)}_{j}\right)} d\varphi_{1}^{(1)}\cdots d\varphi_{N}^{(K)},
\end{split}
\end{equation}
where $N=2^n$. The second equality is due to the symmetries of $\cos \varphi_j^{(k)}$ within $0\sim 2\pi$. Note that this integral cannot be computed analytically. However, as shown in~\cite{chen2020quantum}, $\mathbb{E}\left[\left( \bra{\psi^{(K)}} H \ket{\psi^{(K)}} \right)^2 \right]$ increases with $K$, leading to a decrease in the average variance $\mathbb{E}\left[ \left(\sigma^{2}\right)^{(K)}  \right]$. Therefore, finally we have:
\begin{equation}
        \mathbb{E}\left[ (\sigma^{2})^{(K)} \right] = \sum_{l}\lambda_l^2\left| c_l^{(0)} \right|^2 - \left(\frac{1}{2\pi} \right)^{NK} \int_{0}^{2\pi}\cdots\int_{0}^{2\pi}  \frac{\left(\sum_{j} \left|c_j\right|^2\lambda_j \prod_{k=1}^{K}\left(1-\cos \varphi^{(k)}_{j}\right)\right)^2}{\sum_{j} \left|c_j\right|^2 \prod_{k=1}^{K} \left(1-\cos \varphi^{(k)}_{j}\right)} d\varphi_{1}^{(1)}\cdots d\varphi_{N}^{(K)}.
\end{equation}
\section{Two-device filter}

The setup of two-device filter has been shown in Fig.~\ref{intro}(b) and Algorithm~\ref{algo1}. Below we analyze the behavior of two-device filter under no, \emph{weak} and \emph{strong} postselection, respectively.

\subsection{Energy}\label{appx-dist-energy}

\subsubsection{No postselection}

In the case where we do not postselect anything and let the two-device protocol run until $K$ iterations, at each iteration $k$, we have multiple possible operations (Kraus operators) that can be applied on the state (we omit the iteration notation $^{(k)}$ for now):
\begin{equation}
    \begin{split}
        \mathbf{E}_{0}^{00}=&\frac{1}{4}(I+U)\otimes (I+U) = \frac{1}{4}\sum_{j,j'}(1+e^{i\varphi_{j}})(1+e^{i\varphi_{j'}})\ket{\phi_{j}\phi_{j'}}\bra{\phi_{j}\phi_{j'}}; \\
         \mathbf{E}_{0}^{01}=&\frac{1}{4}(I\otimes I-U\otimes U) = \frac{1}{4}\sum_{j,j'}(1-e^{i(\varphi_{j}+\varphi_{j'})})\ket{\phi_{j}\phi_{j'}}\bra{\phi_{j}\phi_{j'}} =\mathbf{E}_{0}^{10};\\
         \mathbf{E}_{0}^{11}=&\frac{1}{4}(I-U)\otimes (I-U)=\frac{1}{4}\sum_{j,j'}(1-e^{i\varphi_{j}})(1-e^{i\varphi_{j'}})\ket{\phi_{j}\phi_{j'}}\bra{\phi_{j}\phi_{j'}}; \\
         \mathbf{E}_{1}^{01}=&\frac{1}{4}(U\otimes I-I\otimes U)= \frac{1}{4}\sum_{j,j'}(e^{i\varphi_{j}}-e^{i\varphi_{j'}})\ket{\phi_{j}\phi_{j'}}\bra{\phi_{j}\phi_{j'}} =-\mathbf{E}_{1}^{10}; \\
         \mathbf{E}_{1}^{00}=&0=\mathbf{E}_{1}^{11}.
    \end{split}
\end{equation}
where the subscript is the measurement result of top auxiliary qubit and the superscript denotes the other two. One can easily check that:
\begin{equation}
    \sum_{j,k,l=0}^{1} {\mathbf{E}_{j}^{kl}}^{\dagger}{\mathbf{E}_{j}^{kl}} = I.
\end{equation}
For energy, similarly, we consider the quantum channel. Initially, the state is:
\begin{equation}
    \rho^{(0)} = \ket{\psi^{(0)}\psi^{(0)}}\bra{\psi^{(0)}\psi^{(0)}}.
\end{equation}
Once we apply the quantum channel with Kraus operators $\mathcal{E}$, we have the state after one round of the protocol as:
\begin{equation}
\begin{split}
    \rho^{(1)} =& \sum_{j,k,l=0}^{1} \mathbf{E}_{j}^{kl}\rho^{(0)}{\mathbf{E}_{j}^{kl}}^{\dagger} \\
    =& \frac{1}{4} \left( \rho^{(0)} + (I\otimes U) \rho^{(0)} (I\otimes U^{\dagger}) + (U \otimes I) \rho^{(0)} (U^{\dagger} \otimes I) +  (U\otimes U) \rho^{(0)} (U^{\dagger}\otimes U^{\dagger})   \right).
\end{split}
\end{equation}
Then we can define our new equivalent channel with Kraus operators:
\begin{equation}
\begin{split}
    \mathcal{K}_1 =& I \otimes I = \sum_{j,j'} \ket{\phi_j \phi_{j'}}\bra{\phi_j \phi_{j'}}; \\
    \mathcal{K}_2 =& I \otimes U = \sum_{j,j'} e^{i\varphi_j} \ket{\phi_j \phi_{j'}}\bra{\phi_j \phi_{j'}}; \\
    \mathcal{K}_3 =& U \otimes I = \sum_{j,j'} e^{i\varphi_{j'}} \ket{\phi_j \phi_{j'}}\bra{\phi_j \phi_{j'}}; \\
    \mathcal{K}_4 =& U \otimes U = \sum_{j,j'} e^{i(\varphi_j+\varphi_{j'})} \ket{\phi_j \phi_{j'}}\bra{\phi_j \phi_{j'}}. \\
\end{split}
\end{equation}
Therefore, we have:
\begin{equation}
    \rho^{(1)} = \frac{1}{4} \sum_{l,l',h,h'}\left(1+e^{i(\varphi_{l}-\varphi_{h})}\right)\left(1+e^{i(\varphi_{l'}-\varphi_{h'})}\right)c_l c_{l'} c_h^{*} c_{h'}^{*} \ket{\phi_{l}\phi_{l'}}\bra{\phi_{h}\phi_{h'}}.
\end{equation}
Here we also denote $c$ as the input population and omit the iteration notation $^{(k)}$. Therefore, after $K$ rounds, we have (adding the superscripts on $\varphi$ to denote the different time evolution for each iteration $k$):
\begin{equation}
    \rho^{(K)}= \frac{1}{4^K} \sum_{l,l',h,h'}\prod_{k=1}^{K}\left(1+e^{i(\varphi^{(k)}_{l}-\varphi^{(k)}_{h})}\right)\left(1+e^{i(\varphi^{(k)}_{l'}-\varphi^{(k)}_{h'})}\right)c_l c_{l'} c_h^{*} c_{h'}^{*} \ket{\phi_{l}\phi_{l'}}\bra{\phi_{h}\phi_{h'}}.
\end{equation}
Now, we trace out one of the system (e.g., Bob) and preserve one (e.g., Alice), we have:
\begin{equation}
    \rho_{A}^{(K)} = \Tr_{B}\left( \rho^{(K)} \right) =\frac{1}{2^K} \sum_{l,h} \sum_{\alpha \in \mathcal{P}(\{1,\cdots,K\})} e^{i\sum_{a\in\alpha}(\varphi_l^{(a)}-\varphi_h^{(a)})}c_l c_h^{*} \ket{\varphi_l}\bra{\varphi_h}.
\end{equation}
Then, the average energy becomes:
\begin{equation}
\begin{split}
\mathbb{E}\left[E^{(K)}\right] =& \mathbb{E}\left[E^{(K)} \Big| \mathbf{t} \right] = \Tr \left( H\rho_{A}^{(K)}  \right)   \\
=& \frac{1}{2^K} \Tr \left( \sum_{j}\lambda_j \ket{\phi_j}\bra{\phi_j} \sum_{l,h} \sum_{\alpha \in \mathcal{P}(1,\cdots,K)} e^{i\sum_{a\in\alpha}(\varphi_l^{(a)}-\varphi_h^{(a)})}c_l c_h^{*} \ket{\varphi_l}\bra{\varphi_h}  \right) \\
=& \frac{1}{2^K} \Tr \left( \sum_{j,l,h} \sum_{\alpha \in \mathcal{P}(1,\cdots,K)} \lambda_j e^{i\sum_{a\in\alpha}(\varphi_l^{(a)}-\varphi_h^{(a)})} c_l c_h^{*} \delta(h,j)\delta(l,j)  \right) \\
=& \frac{1}{2^K} \Tr \left(  \sum_{j} \lambda_j 2^K \left| c_j  \right|^2  \right) = \sum_{j} \lambda_j\left| c_j  \right|^2.
\end{split}
\end{equation}
This concludes that in the case where there is no postselection, the average energy will stay the same as the energy of the initial input state $\ket{\psi^{(0)}}$.

\subsubsection{\emph{Weak} postselection} \label{app:ssec:weak}

Now we throw away the cases with $\ket{1}$ on the top auxiliary qubit. In this case, after one iteration of the protocol, we have the density matrix:
\begin{equation}
\rho^{(1)} = \frac{1}{\Tr(\cdot)}\left(\rho^{(0)} + \frac{I \otimes U + U \otimes I}{\sqrt{2}}\rho^{(0)}\frac{I \otimes U^{\dagger} + U^{\dagger} \otimes I}{\sqrt{2}} + U\otimes U  \rho^{(0)}U^{\dagger}\otimes U^{\dagger}  \right),
\end{equation}
where $\Tr(\cdot)$ is the trace normalization term to make sure $\Tr(\rho)=1$.

We define:
\begin{equation}
\begin{split}
    \mathcal{K}_1 =& I \otimes I; \\
    \mathcal{K}_2 =& \frac{I \otimes U + U \otimes I}{\sqrt{2}}; \\
    \mathcal{K}_3 =& U \otimes U. \\
\end{split}
\end{equation}
Therefore, after $K$ rounds, we have (note that $\mathcal{K}$ changes for each iteration due to different time evolution $t_k$):
\begin{equation}
\begin{split}
&\rho^{(K)} \\
=&\frac{1}{\Tr(\cdot)}\sum_{j,j',h,h'}\prod_{k=1}^{K}\left( 1+\frac{\left( e^{i\varphi_{j}^{(k)}}+e^{i\varphi_{j'}^{(k)}} \right)\left( e^{-i\varphi_{h}^{(k)}}+e^{-i\varphi_{h'}^{(k)}} \right)}{2} + e^{i\varphi_{j}^{(k)}}e^{i\varphi_{j'}^{(k)}} e^{-i\varphi_{h}^{(k)}}e^{-i\varphi_{h'}^{(k)}} \right) \\
&c_j c_{j'} c_h^{*} c_{h'}^{*}\ket{\phi_j\phi_{j'}}\bra{\phi_h\phi_{h'}}.
\end{split}
\end{equation}
Similarly, we trace out one of the parties and then we have:
\begin{equation}
    \rho_{A}^{(K)} = \frac{\sum_{j,j',h}\prod_{k=1}^{K}\left( 1+\frac{\left( e^{i\varphi_{j}^{(k)}}+e^{i\varphi_{j'}^{(k)}} \right)\left( e^{-i\varphi_{h}^{(k)}}+e^{-i\varphi_{j'}^{(k)}} \right)}{2} + e^{i\varphi_{j}^{(k)}} e^{-i\varphi_{h}^{(k)}} \right)c_j \left|c_{j'}\right|^2 c_h^{*} \ket{\phi_j}\bra{\phi_h}}{\sum_{j,j'}\prod_{k=1}^{K}\left( 3+ \cos{t_k(\lambda_j-\lambda_{j'})}  \right)\left|c_{j}\right|^2 \left|c_{j'}\right|^2  }.
\end{equation}
Therefore:
\begin{equation}
    E^{(K)} = \Tr\left( H\rho_A^{(K)} \right) =\frac{\sum_{j,j'}\lambda_j\prod_{k=1}^{K}\left( 1+\frac{1}{3}\cos(\varphi_j^{(k)}-\varphi_{j'}^{(k)})  \right)\left|c_{j}\right|^2 \left|c_{j'}\right|^2}{\sum_{j,j'}\prod_{k=1}^{K}\left( 1+\frac{1}{3}\cos(\varphi_j^{(k)}-\varphi_{j'}^{(k)})  \right)\left|c_{j}\right|^2 \left|c_{j'}\right|^2}.
\end{equation}
and:
\begin{equation}
\mathbb{E}\left[ E^{(K)} \right] = \left(\frac{1}{2\pi}\right)^{NK} \int_{0}^{2\pi} \cdots \int_{0}^{2\pi} \frac{\sum_{j,j'}\lambda_j\prod_{k=1}^{K}\left( 1+\frac{1}{3}\cos(\varphi_j^{(k)}-\varphi_{j'}^{(k)})  \right)\left|c_{j}\right|^2 \left|c_{j'}\right|^2}{\sum_{j,j'}\prod_{k=1}^{K}\left( 1+\frac{1}{3}\cos(\varphi_j^{(k)}-\varphi_{j'}^{(k)})  \right)\left|c_{j}\right|^2 \left|c_{j'}\right|^2} d\varphi_{1}^{(1)}\cdots d\varphi_{N}^{(K)}.
\label{integral_E_weak}
\end{equation}
This integral is also hard to be solved analytically. However, let us make the following observation:
\begin{proposition}
\label{prop_E_weak}
    Eq.~\eqref{integral_E_weak} can be approximated as:
    \begin{equation}
    \begin{split}
        &\left(\frac{1}{2\pi}\right)^{NK} \int_{0}^{2\pi} \cdots \int_{0}^{2\pi} \frac{\sum_{j,j'}\lambda_j\prod_{k=1}^{K}\left( 1+\frac{1}{3}\cos(\varphi_j^{(k)}-\varphi_{j'}^{(k)})  \right)\left|c_{j}\right|^2 \left|c_{j'}\right|^2}{\sum_{j,j'}\prod_{k=1}^{K}\left( 1+\frac{1}{3}\cos(\varphi_j^{(k)}-\varphi_{j'}^{(k)})  \right)\left|c_{j}\right|^2 \left|c_{j'}\right|^2} d\varphi_{1}^{(1)}\cdots d\varphi_{N}^{(K)} \\
        \approx & \frac{\left(\frac{1}{2\pi}\right)^{NK} \int_{0}^{2\pi} \cdots \int_{0}^{2\pi} \sum_{j,j'}\lambda_j\prod_{k=1}^{K}\left( 1+\frac{1}{3}\cos(\varphi_j^{(k)}-\varphi_{j'}^{(k)})  \right)\left|c_{j}\right|^2 \left|c_{j'}\right|^2d\varphi_{1}^{(1)}\cdots d\varphi_{N}^{(K)}}{\left(\frac{1}{2\pi}\right)^{NK} \int_{0}^{2\pi} \cdots \int_{0}^{2\pi}\sum_{j,j'}\prod_{k=1}^{K}\left( 1+\frac{1}{3}\cos(\varphi_j^{(k)}-\varphi_{j'}^{(k)})  \right)\left|c_{j}\right|^2 \left|c_{j'}\right|^2d\varphi_{1}^{(1)}\cdots d\varphi_{N}^{(K)}} \\
    =& \frac{\mathbb{E}\left[ \sum_{j,j'}\lambda_j\prod_{k=1}^{K}\left( 1+\frac{1}{3}\cos(\varphi_j^{(k)}-\varphi_{j'}^{(k)})  \right)\left|c_{j}\right|^2 \left|c_{j'}\right|^2 \right]}{\mathbb{E}\left[ \sum_{j,j'}\prod_{k=1}^{K}\left( 1+\frac{1}{3}\cos(\varphi_j^{(k)}-\varphi_{j'}^{(k)})  \right)\left|c_{j}\right|^2 \left|c_{j'}\right|^2 \right]} \\
    =& \frac{\left(\left(\frac{4}{3}\right)^K-1\right)\sum_{j}\lambda_j\left| c_j \right|^4+\sum_{j}\lambda_j\left| c_j \right|^2}{\left(\left(\frac{4}{3}\right)^K-1\right)\sum_{j}\left| c_j \right|^4+1}.
    \end{split}
    \end{equation}
\end{proposition}
\begin{proof}
    Let's denote:
    \begin{equation}
        g(\boldsymbol{\varphi}) = \sum_{j,j'}\prod_{k=1}^{K}\left( 1+\frac{1}{3}\cos(\varphi_j^{(k)}-\varphi_{j'}^{(k)})  \right)\left|c_{j}\right|^2 \left|c_{j'}\right|^2.
    \end{equation}
    In order to make the approximation accurate enough, we are going to show that the variance of $g(\boldsymbol{\varphi})$:
    \begin{equation}
        \sigma^2\left[ g(\boldsymbol{\varphi}) \right] = \mathbb{E}\left[ g^2(\boldsymbol{\varphi}) \right] - \left(\mathbb{E}\left[ g(\boldsymbol{\varphi}) \right]\right)^2,
    \end{equation}
    is sufficiently small compared to $\mathbb{E}\left[ g(\boldsymbol{\varphi}) \right]$. Since:
    \begin{equation}
    \mathbb{E}\left[ g(\boldsymbol{\varphi})  \right] = \left(\left(\frac{4}{3}\right)^K-1\right)\sum_{j}\left| c_j \right|^4+1,
    \end{equation}
    and:
    \begin{equation}
    \begin{split}
\mathbb{E}\left[ g^2(\boldsymbol{\varphi})  \right] =& \left(\left( \frac{4}{3} \right)^K \sum_{j} \left| c^{(0)}_j \right|^4\right)^2 + 2\left( \frac{4}{3} \right)^K \left(\sum_{j} \left| c^{(0)}_j \right|^4\right) \times \left(\sum_{j \neq j'} \left| c_j \right|^2 \left| c_{j'} \right|^2\right) + \\
&\sum_{\substack{
    \{j\},\{j'\},\{l\},\{l'\}\text{ and} \\
    \{j,l\},\{j'\},\{l'\}\text{ and} \\
    \{j,l'\},\{j'\},\{l\}\text{ and} \\
    \{j',l\},\{j\},\{l'\}\text{ and} \\
    \{j',l'\},\{j\},\{l\}}}
    \left| c_j \right|^2 \left| c_{j'} \right|^2\left| c_l \right|^2 \left| c_{l'} \right|^2 + \left(\frac{19}{18}\right)^K \sum_{\substack{\{j,l\},\{j',l'\}\text{ and} \\ \{j,l'\}, \{j',l\} }} \left| c_j \right|^2 \left| c_{j'} \right|^2 \left| c_l \right|^2 \left| c_{l'} \right|^2, 
\end{split}
\end{equation}
where the indices inside each braces are equal to each other, and are completely distinct with indices in other braces. For example, $\{j,l\},\{j'\},\{l'\}$ denotes $j=l\neq j' \neq l' \neq j=l$. Then, we have:
\begin{equation}
\begin{split}
\sigma^2\left[g(\boldsymbol{\varphi})\right] =& \mathbb{E}\left[ g^2(\boldsymbol{\varphi})  \right] - \mathbb{E}^2\left[ g(\boldsymbol{\varphi})  \right] \\
=& 2 \left(\left(\frac{19}{18}\right)^K-1\right) \left( \left(\sum_{j} \left| c_j \right|^4\right)^2-\sum_{j} \left| c_j \right|^8   \right) \\
\ll& \left(\left(\frac{4}{3}\right)^K-1\right)\sum_{j}\left| c_j \right|^4+1 = \mathbb{E}\left[ g(\boldsymbol{\varphi})  \right].
\end{split}
\end{equation}
Therefore, we approximate this integral of quotient into quotient of integrals, which does not hold in general but holds for this specific case (or any other similar integral forms).
\end{proof}

Therefore, the average energy becomes:
\begin{equation}
\mathbb{E}\left[ E^{(K)} \right] \approx \frac{\left(\left(\frac{4}{3}\right)^K-1\right)\sum_{j}\lambda_j\left| c_j \right|^4+\sum_{j}\lambda_j\left| c_j \right|^2}{\left(\left(\frac{4}{3}\right)^K-1\right)\sum_{j}\left| c_j \right|^4+1},
\end{equation}
which introduces energy biases to the system. Moreover, the biases are bounded, as $K\rightarrow\infty$, we have:
\begin{equation}
    \mathbb{E}\left[ E^{(\infty)} \right] \rightarrow \left(\sum_j \lambda_j \left| c_j^{(0)} \right|^4\right) \Big/ \left(\sum_j \left| c_j^{(0)} \right|^4\right).
\end{equation}

We can further analyze cases where the energy distribution of the Hamiltonian $H$ is continuous, extending the applicability to a broader range of scenarios. 
We assume $H$ is traceless and local. This kind of $H$ has Gaussian density of states centered at zero~\cite{hartmann2005spectral}. Also, the input product state $\ket{\psi^{(0)}}$ has populations with Gaussian shape in the eigenbasis~\cite{rai2024matrix}. Therefore, we have the density of states:
\begin{equation}
    D(\lambda)=\frac{1}{\sqrt{2\pi\sigma^2}}\exp\left( -\frac{\lambda^2}{2\sigma^2} \right),
\end{equation}
and the populations:
\begin{equation}
    A(\lambda) = \mathbf{A}\exp\left( -\frac{(\lambda-\mu)^2}{2\xi^2} \right),
\end{equation}
where $\sigma^2$ and $\xi^2$ denote the variance of density of states and populations, respectively. $\mu$ is the mean of the Gaussian shape (not the mean of the overall populations) and $\mathbf{A}$ is the normalization factor. Firstly, we need to find out $\mathbf{A}$, as:
\begin{equation}
     \int_{-\infty}^{\infty} D(\lambda) A(\lambda) d\lambda = \mathbf{A}\int_{-\infty}^{\infty} \frac{1}{\sqrt{2\pi\sigma^2}}\exp\left( -\frac{\lambda^2}{2\sigma^2} \right) \exp\left( -\frac{(\lambda-\mu)^2}{2\xi^2} \right) d\lambda = 1.
\end{equation}
Therefore:
\begin{equation}
    \mathbf{A} = \frac{\exp\left( \frac{\mu^2}{2(\xi^2+\sigma^2)} \right) \sqrt{\xi^2+\sigma^2}}{\xi}.
\end{equation}
Then:
\begin{equation}
    \sum_{j} \left| c^{(0)}_j \right|^4 \approx \int_{-\infty}^{\infty} A^2(\lambda) D(\lambda) d\lambda = \frac{\exp\left(\frac{\mu^2}{\xi^2+\sigma^2}-\frac{\mu^2}{\xi^2+2\sigma^2}\right)(\xi^2+\sigma^2)}{\xi\sqrt{\xi^2+2\sigma^2}},
\label{eq:c4-continuous}
\end{equation}
\begin{equation}
    \sum_{j} \lambda_j \left| c^{(0)}_j \right|^4 \approx \int_{-\infty}^{\infty}\lambda A^2(\lambda) D(\lambda) d\lambda = \frac{2\exp\left(\frac{\mu^2}{\xi^2+\sigma^2}-\frac{\mu^2}{\xi^2+2\sigma^2}\right)\mu\sigma^2(\xi^2+\sigma^2)}{\xi(\xi^2+2\sigma^2)^{\frac{3}{2}}},
\end{equation}
\begin{equation}
    E^{(0)}=\sum_{j} \lambda_j \left| c^{(0)}_j \right|^2 \approx \int_{-\infty}^{\infty}\lambda A(\lambda) D(\lambda) d\lambda = \frac{\sigma^2}{\xi^2+\sigma^2}\mu.
\end{equation}
Therefore:
\begin{equation}
    \mathbb{E}\left[ E^{(K)} \right] = \mu\frac{ \sigma^2 \left( 1 + \frac{2 \left(-1 + \left(\frac{4}{3}\right)^K\right) e^{\frac{\mu^2 \sigma^2}{\xi^4 + 3 \xi^2 \sigma^2 + 2 \sigma^4}} \left(\xi^2 + \sigma^2\right)^2}{\xi \left(\xi^2 + 2 \sigma^2\right)^{3/2}} \right)}{\left(\xi^2 + \sigma^2\right) \left( 1 + \frac{\left(-1 + \left(\frac{4}{3}\right)^K\right) e^{\frac{\mu^2 \sigma^2}{\xi^4 + 3 \xi^2 \sigma^2 + 2 \sigma^4}} \left(\xi^2 + \sigma^2\right)}{\xi \sqrt{\xi^2 + 2 \sigma^2}} \right)}.
\end{equation}
The energy goes up for positive $\mu$ and goes down for negative $\mu$. Then, for $K\rightarrow \infty$, we have:
\begin{equation}
\mathbb{E}\left[ E^{(\infty)} \right] = \frac{\mu}{\frac{\xi^2}{2\sigma^2}+1}.
\end{equation}
Then the largest bias we can introduce is approximately:
\begin{equation}
    \left|E^{(0)}-\mathbb{E}\left[ E^{(\infty)} \right]\right| =\left|\frac{\mu}{\frac{\xi^2}{2\sigma^2}+1} -\frac{\mu}{\frac{\xi^2}{\sigma^2}+1}\right| = |\mu|\left( \frac{1}{\frac{\xi^2}{2\sigma^2}+1} -\frac{1}{\frac{\xi^2}{\sigma^2}+1} \right)\leq |\mu| \frac{1}{2\sqrt{2}+3} \approx 0.1716 |\mu|.
\end{equation}
And there is no energy bias if $\mu=0$.

\subsubsection{\emph{Strong} postselection}

Now, if we have \emph{strong} postselection, i.e., we only preserve the one with $\ket{000}$ and $\ket{011}$ on all three auxiliary qubits. We have:
\begin{equation}
    \rho^{(1)} = \frac{1}{\Tr(\cdot)} \left(  (I\otimes U + U\otimes I) \rho^{(0)} (I\otimes U^{\dagger} + U^{\dagger}\otimes I) + (I\otimes I + U\otimes U) \rho^{(0)} (I\otimes I + U^{\dagger}\otimes U^{\dagger}) \right).
\end{equation}
We define:
\begin{equation}
\begin{split}
    \mathcal{K}_1 &= I\otimes U + U\otimes I; \\
    \mathcal{K}_2 &= I\otimes I + U\otimes U.
\end{split}
\end{equation}
Therefore, after $K$ rounds, we have (note that $\mathcal{K}$ changes for each round due to different time evolution $t_k$):
\begin{equation}
\begin{split}
\rho^{(K)} = \frac{1}{\Tr(\cdot)} \sum_{j,j',h,h'} & \prod_{k=1}^{K} \left[ \left(e^{i\varphi_{j}^{(k)}}+e^{i\varphi_{j'}^{(k)}} \right)\left(e^{-i\varphi_{h}^{(k)}}+e^{-i\varphi_{h'}^{(k)}} \right) + (1+e^{i(\varphi_{j}^{(k)}+\varphi_{j'}^{(k)})})(1+e^{-i(\varphi_{h}^{(k)}+\varphi_{h'}^{(k)})}) \right] \\
&c_j c_{j'} c_h^{*} c_{h'}^{*}\ket{\phi_j \phi_{j'}} \bra{\phi_h \phi_{h'}}
\end{split}
\end{equation}
Once we trace out one of the systems, we have:
\begin{equation}
\small
\begin{split}
&\rho^{(K)}_{A} \\
=& \frac{\sum_{j,j',h}\prod_{k=1}^{K} \left[ \left(e^{i\varphi_{j}^{(k)}}+e^{i\varphi_{j'}^{(k)}} \right)\left(e^{-i\varphi_{h}^{(k)}}+e^{-i\varphi_{j'}^{(k)}} \right) + (1+e^{i(\varphi_{j}^{(k)}+\varphi_{j'}^{(k)})})(1+e^{-i(\varphi_{h}^{(k)}+\varphi_{j'}^{(k)})}) \right] c_j \left|c_{j'}\right|^2 c_h^{*} \ket{\phi_j} \bra{\phi_h}}{\sum_{j,j'}\prod_{k=1}^{K}\left[ \left(e^{i\varphi_{j}^{(k)}}+e^{i\varphi_{j'}^{(k)}} \right)\left(e^{-i\varphi_{j}^{(k)}}+e^{-i\varphi_{j'}^{(k)}} \right) + (1+e^{i(\varphi_{j}^{(k)}+\varphi_{j'}^{(k)})})(1+e^{-i(\varphi_{j}^{(k)}+\varphi_{j'}^{(k)})})\right]\left|c_{j}\right|^2 \left|c_{j'}\right|^2   }
\end{split}
\end{equation}
Therefore:
\begin{equation}
E^{(K)} = \Tr\left( H\rho_A^{(K)} \right) = \frac{\sum_{j,j'}\lambda_j\prod_{k=1}^{K}\left( 1+\cos \varphi_j^{(k)}\cos \varphi_{j'}^{(k)}  \right)\left|c_{j}\right|^2 \left|c_{j'}\right|^2}{\sum_{j,j'}\prod_{k=1}^{K}\left( 1+\cos \varphi_j^{(k)}\cos \varphi_{j'}^{(k)}  \right)\left|c_{j}\right|^2 \left|c_{j'}\right|^2}
\end{equation}
and:
\begin{equation}
\mathbb{E}\left[ E^{(K)} \right] = \left( \frac{1}{2\pi}  \right)^{NK}\int_{0}^{2\pi}\cdots\int_{0}^{2\pi} \frac{\sum_{j,j'}\lambda_j\prod_{k=1}^{K}\left( 1+\cos \varphi_j^{(k)}\cos \varphi_{j'}^{(k)}  \right)\left|c_{j}\right|^2 \left|c_{j'}\right|^2}{\sum_{j,j'}\prod_{k=1}^{K}\left( 1+\cos \varphi_j^{(k)}\cos \varphi_{j'}^{(k)}  \right)\left|c_{j}\right|^2 \left|c_{j'}\right|^2} d\varphi_1^{(1)}\cdots d\varphi_N^{(K)}
\label{integral_E_strong}
\end{equation}
This is still not analytically solvable. However, our numerical studies show that the approximation similar to Proposition~\ref{prop_E_weak} still holds for all $K\in\mathbb{Z}^{+}$:
\begin{conjecture}
\label{prop_E_strong}
Eq.~\eqref{integral_E_strong} can be approximated as:
\begin{equation}
\begin{split}
    &\left( \frac{1}{2\pi}  \right)^{NK}\int_{0}^{2\pi}\cdots\int_{0}^{2\pi} \frac{\sum_{j,j'}\lambda_j\prod_{k=1}^{K}\left( 1+\cos \varphi_j^{(k)}\cos \varphi_{j'}^{(k)}  \right)\left|c_{j}\right|^2 \left|c_{j'}\right|^2}{\sum_{j,j'}\prod_{k=1}^{K}\left( 1+\cos \varphi_j^{(k)}\cos \varphi_{j'}^{(k)}  \right)\left|c_{j}\right|^2 \left|c_{j'}\right|^2} d\varphi_1^{(1)}\cdots d\varphi_N^{(K)} \\
   \approx&\frac{\left( \frac{1}{2\pi}  \right)^{NK}  \int_{0}^{2\pi}\cdots\int_{0}^{2\pi}  \sum_{j,j'}\lambda_j\prod_{k=1}^{K}\left( 1+\cos \varphi_j^{(k)}\cos \varphi_{j'}^{(k)}  \right)\left|c_{j}\right|^2 \left|c_{j'}\right|^2  d\varphi_1^{(1)}\cdots d\varphi_N^{(K)}}{\left( \frac{1}{2\pi}  \right)^{NK} \int_{0}^{2\pi}\cdots\int_{0}^{2\pi}   \sum_{j,j'}\prod_{k=1}^{K}\left( 1+\cos \varphi_j^{(k)}\cos \varphi_{j'}^{(k)}  \right)\left|c_{j}\right|^2 \left|c_{j'}\right|^2 d\varphi_1^{(1)}\cdots d\varphi_N^{(K)}} \\
   =&\frac{\mathbb{E} \left[ \sum_{j,j'}\lambda_j\prod_{k=1}^{K}\left( 1+\cos \varphi_j^{(k)}\cos \varphi_{j'}^{(k)}  \right)\left|c_{j}\right|^2 \left|c_{j'}\right|^2 \right]}{\mathbb{E} \left[ \sum_{j,j'}\prod_{k=1}^{K}\left( 1+\cos \varphi_j^{(k)}\cos \varphi_{j'}^{(k)}  \right)\left|c_{j}\right|^2 \left|c_{j'}\right|^2 \right]} \\
   =&\frac{\left(\left(\frac{3}{2}\right)^K-1\right)\sum_{j}\lambda_j\left| c_j \right|^4+\sum_{j}\lambda_j\left| c_j \right|^2}{\left(\left(\frac{3}{2}\right)^K-1\right)\sum_{j}\left| c_j \right|^4+1}.
\end{split}
\end{equation}
\end{conjecture}

This conjecture is numerically tested and the approximation matches with the data shown in Fig.~\ref{energy}. 
However, the similar proof of Proposition~\ref{prop_E_weak} does not apply here.
The reason is as follows.
We let:
    \begin{equation}
        g(\boldsymbol{\varphi}) = \left( \frac{1}{2\pi}  \right)^{NK} \int_{0}^{2\pi}\cdots\int_{0}^{2\pi}   \sum_{j,j'}\prod_{k=1}^{K}\left( 1+\cos \varphi_j^{(k)}\cos \varphi_{j'}^{(k)}  \right)\left|c_{j}\right|^2 \left|c_{j'}\right|^2 d\varphi_1^{(1)}\cdots d\varphi_N^{(K)}.
    \end{equation}
Since:
\begin{equation}
    \mathbb{E}\left[ g(\boldsymbol{\varphi}) \right] = \left(\left(\frac{3}{2}\right)^K-1\right)\sum_{j}\left| c_j \right|^4+1,
\end{equation}
and:
\begin{equation}
\begin{split}
\mathbb{E}\left[ g^2(\boldsymbol{\varphi})  \right] =& \left(\frac{19}{8}\right)^K \sum_{j} \left| c_j  \right|^8 + \left(\frac{9}{4}\right)^K \sum_{j\neq j'}\left| c_j  \right|^4\left| c_j  \right|^4 + 2\left( \frac{3}{2} \right)^K \sum_{j,l\neq l'} \left| c_j \right|^4 \left| c_l \right|^2 \left| c_{l'} \right|^2 + \\
&\sum_{\substack{
    \{j\},\{j'\},\{l\},\{l'\}\text{ and} \\
    \{j,l\},\{j'\},\{l'\}\text{ and} \\
    \{j,l'\},\{j'\},\{l\}\text{ and} \\
    \{j',l\},\{j\},\{l'\}\text{ and} \\
    \{j',l'\},\{j\},\{l\}}}
    \left| c_j \right|^2 \left| c_{j'} \right|^2\left| c_l \right|^2 \left| c_{l'} \right|^2 + \left(\frac{5}{4}\right)^K \sum_{\substack{\{j,l\},\{j',l'\}\text{ and} \\ \{j,l'\}, \{j',l\} }} \left| c_j \right|^2 \left| c_{j'} \right|^2 \left| c_l \right|^2 \left| c_{l'} \right|^2. 
\end{split}
\end{equation}
Therefore:
\begin{equation}
\begin{split}
    \sigma^2\left[g(\boldsymbol{\varphi})\right] =& \mathbb{E}\left[ g^2(\boldsymbol{\varphi})  \right] - \mathbb{E}^2\left[ g(\boldsymbol{\varphi})  \right] \\
    =& 2\left( \left( \frac{5}{4} \right)^K-1 \right) \left(\sum_{j}\left| c_j \right|^4\right)^2 + \left(   \left(\frac{19}{8}\right)^K - \left(\frac{9}{4}\right)^K -2 \left(\frac{5}{4}\right)^K +2 \right)\sum_{j}\left| c_j \right|^8. 
\end{split}
\end{equation}
Note that $\sigma^2\left[g(\boldsymbol{\varphi})\right]\ll \mathbb{E}\left[ g^2(\boldsymbol{\varphi})  \right]$ only holds for small $K$. For large $K$, $\sigma^2\left[g(\boldsymbol{\varphi})\right]\gg \mathbb{E}\left[ g^2(\boldsymbol{\varphi})  \right]$. 
This is why a similar analysis of Proposition~\ref{prop_E_weak} does not apply here.
Moreover, direct samples of Eq.~\eqref{integral_E_strong} to compute the integral may have precision problem, as the number of random $\varphi$ to sample is $2^{NK}$, scaled exponentially for both $N$ and $K$.
Again, if this approximation holds, we can proceed to a similar analysis for continuous case, and get the same energy bias bound in the end.

\subsection{Energy Variance}\label{appendix:energy-variance}

In two-device case, we trace out one of the systems in the output state $\ket{\Psi^{(K)}}$ at iteration $K$ and we have:
\begin{equation}
    \mathbb{E}\left[ \left(\sigma^{2}\right)^{(K)}  \right] = \mathbb{E}\left[\bra{\Psi^{(K)}} I \otimes H^2\ket{\Psi^{(K)}}\right] - \mathbb{E}\left[\left( \bra{\Psi^{(K)}}I \otimes H\ket{\Psi^{(K)}} \right)^2\right].
\end{equation}
Let's analyze the first term first. We denote $\mathbb{E}_{1}$, $\mathbb{E}_{\mathrm{n}}$, $\mathbb{E}_{\mathrm{w}}$ and $\mathbb{E}_{\mathrm{s}}$ as the expected values for cases of single-device, two-device without postselection, with \emph{weak} postselection, and with \emph{strong} postselection, respectively. Therefore, from the previous analysis, and also the Proposition~\ref{prop_E_weak} and~\ref{prop_E_strong}, we have:
\begin{equation}
    \mathbb{E}_{1}\left[\bra{\psi^{(K)}}  H^2\ket{\psi^{(K)}}\right] = \mathbb{E}_{\mathrm{n}}\left[\bra{\Psi^{(K)}}  H^2\ket{\Psi^{(K)}}\right] =\sum_{j}\lambda^2_j\left| c_j \right|^2,
\end{equation}
\begin{equation}
    \mathbb{E}_{\mathrm{w}}\left[\bra{\Psi^{(K)}}  I\otimes H^2\ket{\Psi^{(K)}}\right] =\frac{\left(\left(\frac{4}{3}\right)^K-1\right)\sum_{j}\lambda^2_j\left| c_j \right|^4+\sum_{j}\lambda^2_j\left| c_j \right|^2}{\left(\left(\frac{4}{3}\right)^K-1\right)\sum_{j}\left| c_j \right|^4+1},
\end{equation}
\begin{equation}
    \mathbb{E}_{\mathrm{s}}\left[\bra{\Psi^{(K)}}  I\otimes H^2\ket{\Psi^{(K)}}\right] =\frac{\left(\left(\frac{3}{2}\right)^K-1\right)\sum_{j}\lambda^2_j\left| c_j \right|^4+\sum_{j}\lambda^2_j\left| c_j \right|^2}{\left(\left(\frac{3}{2}\right)^K-1\right)\sum_{j}\left| c_j \right|^4+1}.
\end{equation}
For the second term, we remind that:
\begin{equation}
    \mathbb{E}_{1}\left[\left(\bra{\psi^{(K)}}  H\ket{\psi^{(K)}}\right)^2\right] = \left(\frac{1}{2\pi} \right)^{N} \int_{0}^{2\pi}\cdots\int_{0}^{2\pi}  \frac{\left(\sum_{j} \left|c_j\right|^2\lambda_j \left(1-\cos \varphi_{j}\right)\right)^2}{\sum_{j} \left|c_j\right|^2  \left(1-\cos \varphi_{j}\right)} d\varphi_{1}\cdots d\varphi_{N}.
\end{equation}
For the two-device case, remind the Kraus operators we have in the no-postselection case:
\begin{equation}
    \begin{split}
        \mathbf{E}_{0}^{00}=&\frac{1}{4}(I+U)\otimes (I+U) = \frac{1}{4}\sum_{j,j'}(1+e^{i\varphi_{j}})(1+e^{i\varphi_{j'}})\ket{\phi_{j}\phi_{j'}}\bra{\phi_{j}\phi_{j'}}; \\
         \mathbf{E}_{0}^{01}=&\frac{1}{4}(I\otimes I-U\otimes U) = \frac{1}{4}\sum_{j,j'}(1-e^{i(\varphi_{j}+\varphi_{j'})})\ket{\phi_{j}\phi_{j'}}\bra{\phi_{j}\phi_{j'}} =\mathbf{E}_{0}^{10};\\
         \mathbf{E}_{0}^{11}=&\frac{1}{4}(I-U)\otimes (I-U)=\frac{1}{4}\sum_{j,j'}(1-e^{i\varphi_{j}})(1-e^{i\varphi_{j'}})\ket{\phi_{j}\phi_{j'}}\bra{\phi_{j}\phi_{j'}}; \\
         \mathbf{E}_{1}^{01}=&\frac{1}{4}(U\otimes I-I\otimes U)= \frac{1}{4}\sum_{j,j'}(e^{i\varphi_{j}}-e^{i\varphi_{j'}})\ket{\phi_{j}\phi_{j'}}\bra{\phi_{j}\phi_{j'}} =-\mathbf{E}_{1}^{10}; \\
         \mathbf{E}_{1}^{00}=&0=\mathbf{E}_{1}^{11}.
    \end{split}
\end{equation}
We can actually write them in a closed form:
\begin{equation}
    \mathbf{E}_{s}^{ab} = \frac{1}{8}\sum_{j,j'} (1+(-1)^a e^{i\varphi_j})(1+(-1)^b e^{i\varphi_{j'}}) + (-1)^s (1+(-1)^b e^{i\varphi_j})(1+(-1)^a e^{i\varphi_{j'}})\ket{\phi_j\phi_{j'}}\bra{\phi_j\phi_{j'}},
\end{equation}
where $a,b,s$ can be either 0 or 1, denoting the auxiliary qubit result $\ket{sab}$ for each round. Now, we apply $\mathbf{E}$ on the pure state $\ket{\psi^{(0)}\psi^{(0)}}$, the state becomes:
\begin{equation}
\begin{split}
    \ket{\Psi^{(1)}} =& \frac{1}{\sqrt{\bra{\psi^{(0)}\psi^{(0)}}(\mathbf{E}^{ab}_{s})^{\dagger}\mathbf{E}_{s}^{ab}\ket{\psi^{(0)}\psi^{(0)}}}} \mathbf{E}_{s}^{ab} \ket{\psi^{(0)}\psi^{(0)}} = \frac{1}{\sqrt{P_s^{ab}}}\mathbf{E}_{s}^{ab} \ket{\psi^{(0)}\psi^{(0)}} \\
    =& \frac{1}{\sqrt{P_{s}^{ab}}} \times \frac{1}{8}\sum_{j,j'} (1+(-1)^a e^{i\varphi_j})(1+(-1)^b e^{i\varphi_{j'}}) + (-1)^s (1+(-1)^b e^{i\varphi_j})(1+(-1)^a e^{i\varphi_{j'}}) \\
    &c_{j} c_{j'}\ket{\phi_j\phi_{j'}},
\end{split}
\label{eq:psi}
\end{equation}
where $P_{s}^{ab}$ denotes the probability of applying operation $\mathbf{E}_{s}^{ab}$ on the state $\ket{\psi^{(0)}\psi^{(0)}}$. Therefore:
\begin{equation}
\begin{split}
     \bra{\Psi^{(1)}}I\otimes H\ket{\Psi^{(1)}} = \frac{1}{P_{s}^{ab}}\frac{1}{64}\sum_{j,j'}\left| c_{j} \right|^2\left| c_{j'} \right|^2 \lambda_{j'} \Theta_{s}^{ab},
\end{split}
\end{equation}
where:
\begin{equation}
    \begin{split}
        \Theta_{0}^{00} =& 16(1+\cos \varphi_{j})(1+\cos \varphi_{j'}) = 64\cos^2\frac{\varphi_{j}}{2}\cos^2\frac{\varphi_{j'}}{2};\\
        \Theta^{01}_{0} =& 8(1-\cos(\varphi_{j}+\varphi_{j'})) = 16\sin^2 \frac{\varphi_{j}+\varphi_{j'}}{2}=\Theta^{10}_{0};\\
        \Theta^{11}_{0} =& 16(1-\cos \varphi_{j})(1-\cos \varphi_{j'}) = 64\sin^2\frac{\varphi_{j}}{2}\sin^2\frac{\varphi_{j'}}{2};\\
        \Theta^{01}_{1} =& 8(1-\cos(\varphi_{j}-\varphi_{j'})) = 16\sin^2 \frac{\varphi_{j}-\varphi_{j'}}{2}=\Theta^{10}_{1};\\
        \Theta^{00}_{1} =& 0 = \Theta^{11}_{1},
    \end{split}
    \label{eq:Theta}
\end{equation}
and:
\begin{equation}
    P_{s}^{ab} = \frac{1}{64}\sum_{j,j'}\left| c_{j} \right|^2\left| c_{j'} \right|^2 \Theta_{s}^{ab}.
\end{equation}
Extending it to $K$ rounds, we have:
\begin{equation}
    \bra{\Psi^{(K)}}I\otimes H \ket{\Psi^{(K)}} = \frac{1}{P^{\mathbf{ab}}_{\mathbf{s}}}\frac{1}{64^K}\sum_{j,j'}|c_j|^2|c_{j'}|^2\lambda_{j'}\prod_{k=1}^{K}\Theta^{a_k b_k}_{s_k},
\end{equation}
and:
\begin{equation}
    P^{\mathbf{ab}}_{\mathbf{s}}=\frac{1}{64^K}\sum_{j,j'}\left| c_{j} \right|^2\left| c_{j'} \right|^2 \prod_{k=1}^{K}\Theta^{a_k b_k}_{s_k},
\end{equation}
with the populations of $\ket{\phi_{j}\phi_{j'}}$:
\begin{equation}
|c^{(K)}_{jj'}|^2=\frac{1}{P^{\mathbf{ab}}_{\mathbf{s}}}\frac{1}{64^K}|c_j|^2|c_{j'}|^2\prod_{k=1}^{K}\Theta^{a_k b_k}_{s_k}.
\end{equation}
Notably, when we pick $\ket{sab}=\ket{000}$ or $\ket{011}$ only (in the \emph{strong} postselection criterion), the states on the devices keep the tensor product structure; that is, $c^{(k)}_{jj'}$ can be always written in $c_j^{(k)}c_{j'}^{(k)}$. This implies that the states of Alice and Bob are pure product states after each iteration of operation.
Then, under the no postselection case:
\begin{equation}
\begin{split}
    \mathbb{E}_{\mathrm{n}}\left[ \left( \bra{\Psi^{(K)}}I\otimes H \ket{\Psi^{(K)}} \right)^2 \right] =& \sum_{\mathbf{a},\mathbf{b},\mathbf{s}}\mathbb{E}\left[ P^{\mathbf{ab}}_{\mathbf{s}}\left( \frac{1}{P^{\mathbf{ab}}_{\mathbf{s}}}\frac{1}{64^K}\sum_{j,j'}|c_j|^2|c_{j'}|^2\lambda_{j'}\prod_{k=1}^{K}\Theta^{a_k b_k}_{s_k} \right)^2 \right] \\
    =& \frac{1}{64^{K}}\sum_{\mathbf{a},\mathbf{b},\mathbf{s}}\mathbb{E}\left[\frac{\left(  \sum_{j,j'}\left| c_{j} \right|^2\left| c_{j'} \right|^2 \lambda_{j'} \prod_{k=1}^{K}\Theta^{a_k b_k}_{s_k} \right)^2}{\sum_{j,j'}\left| c_{j} \right|^2\left| c_{j'} \right|^2  \prod_{k=1}^{K}\Theta^{a_k b_k}_{s_k}}\right].
\end{split}
\end{equation}
Then, for the \emph{weak} postselection case, a normalization on $P^{\mathbf{ab}}_{\mathbf{s}}$ is required. Therefore:
\begin{equation}
\begin{split}
    W^{\mathbf{a'b'}}=\frac{P^{\mathbf{a'b'}}_{\mathbf{0}}}{\sum_{\mathbf{a},\mathbf{b}}P^{\mathbf{ab}}_{\mathbf{0}}}=\frac{P^{\mathbf{a'b'}}_{\mathbf{0}}}{\frac{1}{4^K}\sum_{j,j'}|c_j|^2 |c_{j'}|^2 \prod_{k=1}^{K}(3+\cos(\varphi_j^{(k)}-\varphi_{j'}^{(k)}))},
\end{split}
\end{equation}
and:
\begin{equation}
\begin{split}
    &\mathbb{E}_{\mathrm{w}}\left[ \left( \bra{\Psi^{(K)}}I\otimes H \ket{\Psi^{(K)}} \right)^2 \right] \\
    =& \sum_{\mathbf{a},\mathbf{b}}\mathbb{E}\left[ W^{\mathbf{ab}}\left( \frac{1}{P^{\mathbf{ab}}_{\mathbf{0}}}\frac{1}{64^K}\sum_{j,j'}|c_j|^2|c_{j'}|^2\lambda_{j'}\prod_{k=1}^{K}\Theta^{a_k b_k}_{0} \right)^2 \right] \\
    =&\frac{1}{16^K}\sum_{\mathbf{a},\mathbf{b}}\mathbb{E}\left[\frac{\left(  \sum_{j,j'}\left| c_{j} \right|^2\left| c_{j'} \right|^2 \lambda_{j'} \prod_{k=1}^{K}\Theta^{a_k b_k}_{0} \right)^2}{\sum_{j,j'}\left| c_{j} \right|^2\left| c_{j'} \right|^2  \prod_{k=1}^{K}\Theta^{a_k b_k}_{0}} \left( \sum_{j,j'}|c_j|^2 |c_{j'}|^2 \prod_{k=1}^{K}(3+\cos(\varphi_j^{(k)}-\varphi_{j'}^{(k)}))\right)^{-1}\right].
\end{split}
\end{equation}
Similarly, for the \emph{strong} postselection, the normalization on $P^{\mathbf{ab}}_{\mathbf{s}}$ is:
\begin{equation}
    S^{a'_k=b'_k}=\frac{P^{a'_k=b'_k}_{\mathbf{0}}}{\sum_{a_k=b_k}P^{a_k=b_k}_{\mathbf{0}}}=\frac{P^{a'_k=b'_k}_{\mathbf{0}}}{\frac{1}{2^K}\sum_{j,j'}|c_j|^2|c_{j'}|^2\prod_{k=1}^{K}(1+\cos(\varphi_{j}^{(k)})\cos(\varphi_{j'}^{(k)}))}
\end{equation}
and:
\begin{equation}
\begin{split}
    &\mathbb{E}_{\mathrm{s}}\left[ \left( \bra{\Psi^{(K)}}I\otimes H \ket{\Psi^{(K)}} \right)^2 \right] \\
    =&\sum_{a_k=b_k}\mathbb{E}\left[ S^{a_k=b_k} \left(  \frac{1} {P^{a_k=b_k}_{\mathbf{0}}}\frac{1}{64^K}\sum_{j,j'}|c_j|^2|c_{j'}|^2\lambda_{j'}\prod_{k=1}^{K}\Theta^{a_k= b_k}_{0}  \right)^2 \right] \\
    =&\frac{1}{32^K}\sum_{a_k=b_k}\mathbb{E}\left[  \frac{\left(  \sum_{j,j'}\left| c_{j} \right|^2\left| c_{j'} \right|^2 \lambda_{j'} \prod_{k=1}^{K}\Theta^{a_k = b_k}_{0} \right)^2}{\sum_{j,j'}\left| c_{j} \right|^2\left| c_{j'} \right|^2  \prod_{k=1}^{K}\Theta^{a_k = b_k}_{0}} 
\left( \sum_{j,j'}|c_j|^2|c_{j'}|^2\prod_{k=1}^{K}(1+\cos(\varphi_{j}^{(k)})\cos(\varphi_{j'}^{(k)})) \right)^{-1} \right].
\end{split}
\end{equation}

These are the analytical expressions for the expectation value of $\left( \bra{\Psi^{(K)}}I\otimes H \ket{\Psi^{(K)}} \right)^2$, which are extremely difficult to compute explicitly. However, a similar line of analysis can be carried out by analogy with the single-device filtering process. As shown in Eq.~\eqref{eq:psi} and~\eqref{eq:Theta}, the overall state shared between Alice and Bob remains symmetric throughout the protocol; that is, the state is invariant under the exchange of indices $j$ and $j'$. Consequently, the reduced states of Alice and Bob are always identical. This implies that, asymptotically, the two-device eigenstate preparation protocol prepares two identical eigenstates, provided that cases with ancillary outcome $\ket{s} = \ket{1}$ are excluded. If such outcome occurs, the populations corresponding to the state $\ket{\phi_j} \otimes \ket{\phi_j}$ may vanish, possibly resulting in a final state composed of superpositions of $\ket{\phi_j} \otimes \ket{\phi_{j'}}$ and $\ket{\phi_{j'}} \otimes \ket{\phi_j}$ for $j \neq j'$.

To build intuition, we first focus on the \emph{strong} postselection case. In this scenario, we apply the same filtering factor $(1 \pm \cos \varphi_j^{(k)})$ to the populations of each party, which is the same the single-device case. Also, the global state remains separable, and the time evolutions on both devices are identical at each iteration. As a result, the dominant (unnormalized) population on two identical eigenstates is amplified as:
\begin{equation}
    |c^{(k)}_j|^2 \cdot |c^{(k)}_j|^2 = |c^{(k-1)}_j|^2 \times (1 \pm \cos \varphi_j^{(k)}) \cdot |c^{(k-1)}_j|^2 \times (1 \pm \cos \varphi_j^{(k)}),
\end{equation}
with the larger associated measurement probability given by:
\begin{equation}
\begin{split}
    \frac{P_{0}^{00 \text{ or }11}}{P_{0}^{00} + P_{0}^{11}} &= \frac{\frac{1}{4} \sum_{j,j'} |c^{(k-1)}_j|^2 |c^{(k-1)}_{j'}|^2 (1 \pm \cos \varphi^{(k)}_j)(1 \pm \cos \varphi^{(k)}_{j'})}{\frac{1}{2} \sum_{j,j'} |c^{(k-1)}_j|^2 |c^{(k-1)}_{j'}|^2 \left(1 + \cos \varphi^{(k)}_j \cos \varphi^{(k)}_{j'}\right)} \\
    &= \frac{\left(1 \pm \sum_{j} |c^{(k-1)}_j|^2 \cos \varphi^{(k)}_j \right)^2}{2 + 2\left( \sum_{j} |c^{(k-1)}_j|^2 \cos \varphi^{(k)}_j \right)^2},
\end{split}
\label{eq:magnify-two}
\end{equation}
when $\pm \cos \varphi_j^{(k)} > 0$. The signs $+$ and $-$ correspond to the postselection outcomes $P_0^{00}$ and $P_0^{11}$, respectively. Note that the resulting probability distribution differs from the single-device case. Since $|c_j^{(k-1)}|^2$ is the dominant population, it is likely that $\pm \sum_j |c^{(k-1)}_j|^2 \cos \varphi_j^{(k)} > 0$. In such cases, it is straightforward to verify that Eq.~\eqref{eq:magnify-two} yields a value larger than or equal to the single-device filtering probability,
\[
P_{0 \text{ or } 1} = \frac{1}{2} \pm \frac{1}{2} \sum_j |c_j^{(k-1)}|^2 \cos \varphi_j^{(k)},
\]
meaning that the probability of amplifying the dominant $|c_j^{(k-1)}|^2$ is higher in the two-device setting. This demonstrates that the two-device protocol with \emph{strong} postselection can outperform the single-device case, albeit at the cost of the exponential postselection overhead.

In the case of \emph{weak} postselection, the accepted outcomes include not only those from the \emph{strong} postselection (i.e., $\ket{sab} = \ket{000}$ and $\ket{011}$), but also additional measurement results such as $\ket{001}$ and $\ket{010}$. These outcomes can generally induce entanglement between Alice and Bob’s subsystems. Nevertheless, similar analytical arguments apply, and numerical evidence suggests that the \emph{weak} postselection still outperforms the single-device algorithm, while maintaining a bounded postselection overhead. As the energy variance decreases with successive successful iterations, the entanglement between the two parties gradually diminishes, and the overall state increasingly approximates a product of two identical eigenstates shared between Alice and Bob.

Note that the $P_{\mathbf{s}}^{\mathbf{ab}}$ is analytically calculable for both \emph{weak} and \emph{strong} postselection scenarios. In the \emph{weak} postselection, the cumulative success rate:
\begin{equation}
\begin{split}
    P_{\mathrm{w}} =& \frac{1}{64^K}\sum_{j,j'}|c^{(0)}_j|^2|c^{(0)}_{j'}|^2 \prod_{k=1}^{K} \frac{1}{4\pi^2}\int_{0}^{2\pi}\int_{0}^{2\pi}(64-32\sin^2\frac{\varphi^{(k)}_j - \varphi^{(k)}_{j'}}{2})d\varphi_{j}^{(k)}d\varphi_{j'}^{(k)} \\
    =& \sum_{j}|c^{(0)}_{j}|^4+\left(\frac{3}{4}\right)^K \sum_{j\neq j'} |c^{(0)}_j|^2 |c^{(0)}_{j'}|^2,
\end{split}
\end{equation}
which is lower-bounded by $\sum_{j}|c_{j}|^4$ with respect to $K$. This matches with the numerical result where after several iterations the cumulative success rate plateaus. Therefore, the lower-bound of $P_{\mathrm{w}}$ only depends on the populations of the initial state in the eigenbasis of the Hamiltonian $H$. In the case of product state as the initial state and local Hamiltonian, $\sum_{j}|c_{j}|^4$ can be then approximated in Eq.~\eqref{eq:c4-continuous} with $\sigma,\xi\sim O(\sqrt{n})$ where $n$ is the number of qubits of the prepared eigenstate~\cite{hartmann2005spectral, rai2024matrix}. Therefore, this lower-bound has scaling of $A\exp(B/n)$ where $A$ and $B$ are coefficients independent from $n$.  

On the other hand, for the \emph{strong} postselection case:
\begin{equation}
\begin{split}
    P_{\mathrm{s}}=&\frac{1}{64^K}\sum_{j,j'}|c^{(0)}_j|^2|c^{(0)}_{j'}|^2\prod_{k=1}^{K}\frac{1}{4\pi^2}\int_{0}^{2\pi}\int_{0}^{2\pi}(64\cos^2\frac{\varphi^{(k)}_j}{2}\cos^2\frac{\varphi^{(k)}_{j'}}{2}+64\sin^2\frac{\varphi^{(k)}_j}{2}\sin^2\frac{\varphi^{(k)}_{j'}}{2})d\varphi_j^{(k)}d\varphi_{j'}^{(k)} \\
    =& \left(\frac{3}{4}\right)^{K}\sum_{j}|c^{(0)}_j|^4+\left(\frac{1}{2}\right)^{K}\sum_{j \neq j'}|c^{(0)}_j|^2|c^{(0)}_{j'}|^2,
\end{split}
\end{equation}
which decreases exponentially with $K$.

\subsection{Eigenvalue spreading (variance of the mixed state)} \label{app:ssec:spread}

Throughout this work, we primarily discuss the average variance of the pure state ensemble. 
Specifically, for each repetition, we let the algorithm to run for several iterations, obtaining the output pure state (provided the algorithm does not restart due to postselection). We then compute the energy variance of this output state and store the result. This process is repeated, and the energy variance values are averaged across all repetitions. This averaged value is what we refer to as the average variance of the pure state ensemble, i.e.,
\begin{equation}
    \mathbb{E}\left[ (\sigma^{2})^{(k)} \right] =\mathbb{E}\left[\bra{\psi^{(k)}}H^2\ket{\psi^{(k)}}\right] - \mathbb{E}\left[\left( \bra{\psi^{(k)}}H\ket{\psi^{(k)}} \right)^2\right]=\Tr\left( \rho^{(k)}H^2 \right) - \mathbb{E}\left[ {E^{(k)}}^2 \right],
\end{equation}
where $E^{(k)}=\bra{\psi^{(k)}}H\ket{\psi^{(k)}}$ denotes the energy of pure state $\psi^{(k)}$ at $k$th iteration. 

Alternatively, we can consider an expected ensemble of pure states returned by the algorithm over many iterations, and consider the ensemble as a mixed state $\rho^{(k)}$. 
This mixed state has an associated energy $E^{(k)}=\Tr\left( \rho^{(k)}H \right)$. We can then define another variance, corresponding to the mixed state $\rho^{(k)}$, denoted as $V^{(k)}$:
\begin{equation}
    V^{(k)}=\Tr\left( \rho^{(k)}H^2 \right) - \mathbb{E}\left[ {E^{(k)}}^2 \right] = \Tr\left( \rho^{(k)}H^2 \right) - \left(\Tr(\rho^{(k)} H)\right)^2.
\end{equation}
We highlight that this is a different quantity from $\mathbb{E}\left[ (\sigma^{2})^{(k)} \right]$. We refer to $V^{(k)}$ as the \emph{eigenvalue spread} to avoid using the term variance for both concepts. The eigenvalue spread is the variance of mixed state $\rho$. 

For the analysis below, we distinguish different cases for the respective postselection criterion.

\subsubsection{No postselection}

In this case, there is no energy bias. Therefore, $V^{(K)}$ stays the same as the energy variance of the input state and does not vary for any $K\in\mathbb{Z}^{+}$, i.e., 
\begin{equation}
    V^{(K)} = \sum_{j}\lambda^2_j\left| c_j \right|^2 - \left(\sum_{j}\lambda_j\left| c_j \right|^2\right)^2.
\end{equation}
\subsubsection{\emph{Weak} postselection} \label{app:ssec:spread:weak}

For the first term, we use the approximation of Proposition~\ref{prop_E_weak} but only change $\lambda_j$ into $\lambda_j^2$, which still makes the approximation (integral of quotient $\rightarrow$ quotient of integrals) valid. And the second term has been introduced in the section discussing the energy biases of \emph{weak} postselection cases. Therefore, here:
\begin{equation}
    V^{(K)} \approx \frac{\left(\left(\frac{4}{3}\right)^K-1\right)\sum_{j}\lambda^2_j\left| c_j \right|^4+\sum_{j}\lambda^2_j\left| c_j \right|^2}{\left(\left(\frac{4}{3}\right)^K-1\right)\sum_{j}\left| c_j \right|^4+1} - \left( \frac{\left(\left(\frac{4}{3}\right)^K-1\right)\sum_{j}\lambda_j\left| c_j \right|^4+\sum_{j}\lambda_j\left| c_j \right|^2}{\left(\left(\frac{4}{3}\right)^K-1\right)\sum_{j}\left| c_j \right|^4+1}  \right)^2, 
\label{eq:weak-int-approx}
\end{equation}
and the bound is:
\begin{equation}
    V^{(\infty)} = \frac{\sum_{j}\lambda_j^2 |c_j|^4}{\sum_{j}|c_j|^4}-\left( \frac{\sum_{j}\lambda_j |c_j|^4}{\sum_{j}|c_j|^4} \right)^2. 
\end{equation}
Similarly, we can also extend this into continuous scenarios. As:
\begin{equation}
    \sum_{j} \lambda^2_j \left| c^{(0)}_j \right|^4 \approx \int_{-\infty}^{\infty}\lambda^2 A^2(\lambda) D(\lambda) d\lambda = \frac{\exp\left(\frac{\mu^2}{\xi^2+\sigma^2}-\frac{\mu^2}{\xi^2+2\sigma^2}\right)\sigma^2(\xi^2+\sigma^2)(\xi^4+4\mu^2\sigma^2+2\xi^2\sigma^2)}{\xi(\xi^2+2\sigma^2)^{\frac{5}{2}}};
\end{equation}
\begin{equation}
    \sum_{j} \lambda^2_j \left| c^{(0)}_j \right|^2 \approx \int_{-\infty}^{\infty}\lambda^2 A(\lambda) D(\lambda) d\lambda = \frac{\sigma^2(\xi^4+\mu^2\sigma^2+\xi^2\sigma^2)}{(\xi^2+\sigma^2)^2},
\end{equation}
then we have:
\begin{equation}
    V^{(K)}=\frac{\sigma^2 }{(\xi^2+\sigma^2)^2 }\cdot \frac{\Gamma+\Xi}{\Theta},
\end{equation}
where:
\begin{equation}
    \Theta = \left( 1 + \frac{\left( -1 + \left( \frac{4}{3} \right)^K \right) e^{\frac{\mu^2 \sigma^2}{\xi^4 + 3 \xi^2 \sigma^2 + 2 \sigma^4}} \left( \xi^2 + \sigma^2 \right)}{\xi \sqrt{\xi^2 + 2 \sigma^2}} \right)^2,
\end{equation}
\begin{equation}
    \Gamma = -\sigma^2 \left( \mu + \frac{2 \left( -1 + \left( \frac{4}{3} \right)^K \right) e^{\frac{\mu^2 \sigma^2}{\xi^4 + 3 \xi^2 \sigma^2 + 2 \sigma^4}} \mu \left( \xi^2 + \sigma^2 \right)^2}{\xi \left( \xi^2 + 2 \sigma^2 \right)^{3/2}} \right)^2,
\end{equation}
\begin{equation}
\begin{split}
    \Xi =& \left( 1 + \frac{ \left( -1 + \left( \frac{4}{3} \right)^K \right) e^{\frac{\mu^2 \sigma^2}{\xi^4 + 3 \xi^2 \sigma^2 + 2 \sigma^4}} \left( \xi^2 + \sigma^2 \right)}{\xi \sqrt{\xi^2 + 2 \sigma^2}} \right) \times \\
    &\left( \xi^4 + (\mu^2 + \xi^2) \sigma^2 + \frac{\left( -1 + \left( \frac{4}{3} \right)^K \right) e^{\frac{\mu^2 \sigma^2}{\xi^4 + 3 \xi^2 \sigma^2 + 2 \sigma^4}} \left( \xi^2 + \sigma^2 \right)^3 \left( \xi^4 + 4 \mu^2 \sigma^2 + 2 \xi^2 \sigma^2 \right)}{\xi \left( \xi^2 + 2 \sigma^2 \right)^{5/2}} \right).
\end{split}
\end{equation}
One can find that $V^{(K)}$ decreases with $K$ and does not depend on whether the input state energy is above or below 0. For the cases when $K\rightarrow\infty$, we have the eigenstate spread decrease:
\begin{equation}
    V^{(0)}-V^{(\infty)}=\xi^2 \frac{1}{\frac{\xi^2}{\sigma^2}+1} - \xi^2 \frac{1}{\frac{\xi^2}{\sigma^2}+2}  =\frac{1}{\left( \frac{\xi^2}{\sigma^2}+1 \right)\left( \frac{\xi^2}{\sigma^2}+2 \right)} \xi^2.
\end{equation}
\subsubsection{\emph{Strong} postselection}

As shown in Proposition~\ref{prop_E_strong}, the approximation:
\begin{equation}
    V^{(K)} \approx \frac{\left(\left(\frac{3}{2}\right)^K-1\right)\sum_{j}\lambda^2_j\left| c_j \right|^4+\sum_{j}\lambda^2_j\left| c_j \right|^2}{\left(\left(\frac{3}{2}\right)^K-1\right)\sum_{j}\left| c_j \right|^4+1} - \left( \frac{\left(\left(\frac{3}{2}\right)^K-1\right)\sum_{j}\lambda_j\left| c_j \right|^4+\sum_{j}\lambda_j\left| c_j \right|^2}{\left(\left(\frac{3}{2}\right)^K-1\right)\sum_{j}\left| c_j \right|^4+1}  \right)^2,
\label{eq:strong-int-approx}
\end{equation}
is conjectured to be true. Moreover, this is also supported by the numerics as in Fig.~\ref{V_approx}. Therefore, we can also analyze \emph{strong} postselection for continuous case and finally get the same bound of the eigenstate spread.

\section{Generalization to $s$ devices}\label{appendix:3-dev}

In the main text, we discuss a possible way to generalize our algorithm to multiple devices, where we use the cycle permutations to replace the SWAP operations. Here we briefly give more analytical details on the extensions.

For $s$-device filter, we write the corresponding state in Fig.~\ref{three-dev} step by step. Assuming the input state is $\sum_{j_1,\cdots,j_s}c_{j_1\cdots j_s}\ket{\phi_{j_1}\cdots j_s}$. Before the controlled-$D$ operator, the state is:

\begin{equation}
    \frac{1}{2^{\frac{s}{2}}s^{\frac{1}{2}}}\sum_{\alpha=0}^{s-1}\ket{\alpha}\otimes\sum_{a_1,\cdots,a_s=0}^{1}\ket{a_1\cdots a_s}\otimes \sum_{j_1,\cdots,j_s}\exp\left( i(a_1\varphi_{j_1}+\cdots+a_s\varphi_{j_s}) \right)c_{j_1\cdots j_s}\ket{\phi_{j_1}\cdots \phi_{j_s}},
\end{equation}

where $\ket{\alpha}$ and $\ket{a_1\cdots a_s}$ denote the top and $s$-device local auxiliary qudit/qubits, respectively. Then we apply a multi-level controlled-$D$ operator $\sum_{q=0}^{s-1}\ket{q}\bra{q}\otimes D^q$, and apply Fourier transform on every auxiliary qudit/qubit. Before the measurement, the state becomes:

\begin{equation}
\begin{split}
    \frac{1}{2^s s}\sum_{\alpha=0}^{s-1}\ket{\alpha}\otimes&\sum_{a_1,\cdots,a_s=0}^{1}\sum_{b_1,\cdots,b_s=0}^{1} (-1)^{a_1b_1+\cdots+a_s b_s}\ket{b_1\cdots b_s} \\
    \otimes&\sum_{j_1,\cdots,j_s}\left[\exp\left( i(a_1\varphi_{j_1}+\cdots+a_s\varphi_{j_s}) \right) + \omega^{\alpha}\exp\left( i(a_1\varphi_{j_2}+\cdots+a_s\varphi_{j_1}) \right) + \cdots + \right. \\
    & \left.\omega^{(s-1)\alpha}\exp\left( i(a_1\varphi_{j_s}+\cdots+a_s\varphi_{j_{s-1}}) \right)  \right] c_{j_1\cdots j_s}\ket{\phi_{j_1}\cdots\phi_{j_s}}.
\end{split}
\end{equation}

One can see that if the top auxiliary $s$-level qudit clicks at $\alpha\neq 0$, the populations for $\ket{\phi_{j}\cdots\phi_{j}}$ become 0. That is why we still need to postselect $\ket{\alpha}=\ket{0}$. 

By measuring all the auxiliary qudit and qubits with outcome $\ket{0}\otimes\ket{b_1\cdots b_s}$, we have the state and the corresponding measurement probability as:

\begin{equation}
\begin{split}
    &\frac{1}{2^s s \sqrt{P_0^{b_1\cdots b_s}}}\sum_{a_1,\cdots,a_s=0}^{1}(-1)^{a_1b_1+\cdots+a_s b_s} \\
    &\sum_{j_1,\cdots,j_s}\left[\exp\left( i(a_1\varphi_{j_1}+\cdots+a_s\varphi_{j_s}) \right) + \omega^{\alpha}\exp\left( i(a_1\varphi_{j_2}+\cdots+a_s\varphi_{j_1}) \right) + \cdots + \right. \\
    & \left.\omega^{(s-1)\alpha}\exp\left( i(a_1\varphi_{j_s}+\cdots+a_s\varphi_{j_{s-1}}) \right)  \right] c_{j_1\cdots j_s}\ket{\phi_{j_1}\cdots\phi_{j_s}},
\end{split}
\label{eq:state-3}
\end{equation}

and,

\begin{equation}
\begin{split}
    P_0^{b_1\cdots b_s}=\frac{1}{s^2}\sum_{j_1,\cdots,j_s}|c_{j_1\cdots j_s}|^2 \left( \sum_{cyc \ \varphi_j}\prod_{l=1}^{s}\left( b_l\sin\frac{\varphi_{j_l}}{2}+(1-b_l)\cos\frac{\varphi_{j_l}}{2} \right) \right)^2,
\end{split}
\end{equation}

respectively. From Eq.~\eqref{eq:state-3}, one can see that due to the symmetry of cyclic permutation, the state for each reduced system will still be identical.

Still, the \emph{strong} postselection case can be analyzed similarly. In this case, $|c_{j_1\cdots j_s}|^2=|c_{j_1}|^2\cdots|c_{j_s}|^2$. Once $\ket{b_1\cdots b_s}=\ket{0\cdots 0}$ or $\ket{b_1\cdots b_s}=\ket{1\cdots 1}$, the dominant population $|c_{j}|^2$ for each corresponding index $j$ will also be multiplied by the same factor $\left(1\pm\cos \varphi_j\right)$. When $\pm\cos \varphi_j>0$, the dominant population will be further amplified with increased probability:

\begin{equation}
    \frac{(1\pm\sum_{j}|c_j|^2\cos \varphi_j)^s}{(1+\sum_{j}|c_j|^2\cos \varphi_j)^s+(1-\sum_{j}|c_j|^2\cos \varphi_j)^s}
\end{equation}

as $\pm\sum_j|c_j|^2\cos \varphi_j$ is more likely to be larger than 0. Note that in this circumstance, this probability increases with number of devices $s$, which means that increasing $s$ will yield faster eignestate convergence, but with more communication costs between parties. Also, one can find that on average, $P_0^{0\cdots0}$ and $P_1^{1\cdots1}$ still exponentially decrease with number of iterations $K$. For \emph{weak} postselection cases, the analytics become more complicated but numerical evidences suggest that 3-device filter still works in the similar way and we expect that it also works for general $s$. However, notice that $D^q$ ($q\neq0$) for non-prime $s$ may not constitute full-cycle permutation, which might affect the convergence rate. This has been briefly mentioned in the main text and also discussed in~\cite{liu2025generalized}. Finally, for \emph{weak} postselection, via calculations, the success probability on average is bounded by:

\begin{equation}
    \sum_{j}|c^{(0)}_j|^{2s}
\end{equation}

\section{Resource cost in the distributed settings}\label{appendix:cost}

\begin{figure}
    \centering
    \includegraphics[width=0.8\linewidth]{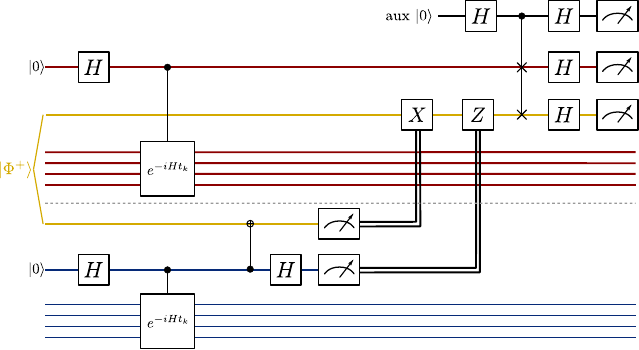}
    \caption{The quantum circuit per iteration to realize the distributed filtering protocol with shared Bell pairs and local operations and classical communication. The qubits above and below the dashed line are held by Alice and Bob, respectively.}
    \label{fig:teleport}
\end{figure}

\begin{figure}
    \centering
    \includegraphics[width=1.0\linewidth]{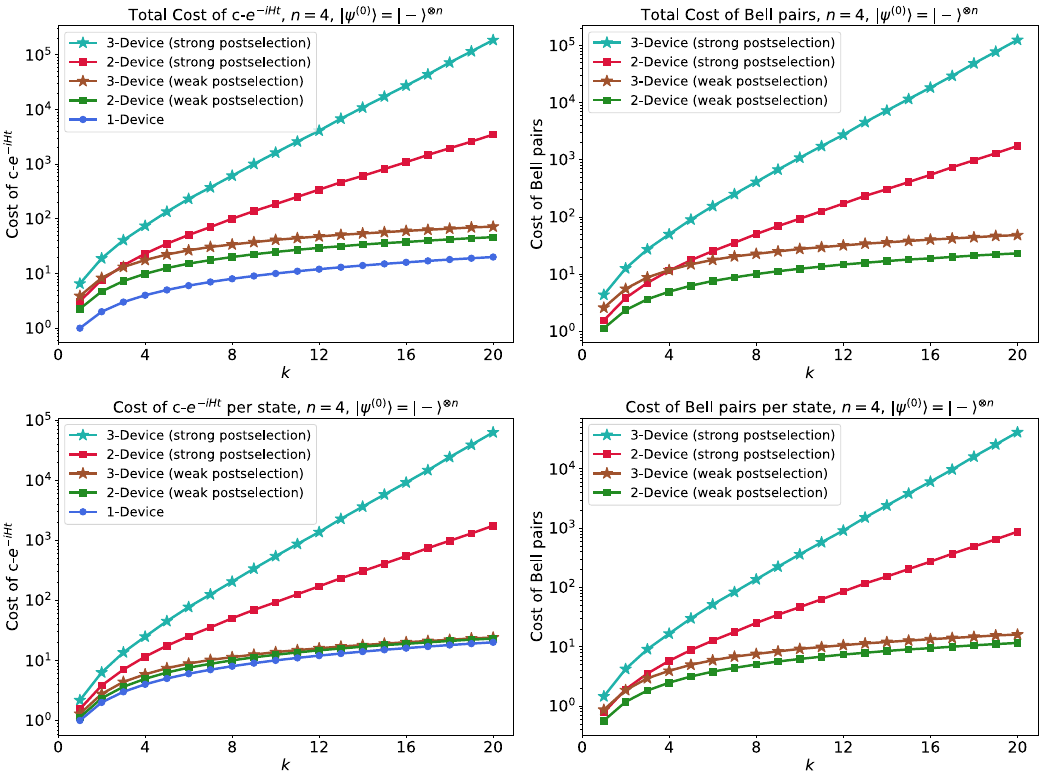}
    \caption{Empirical average number of control time evolutions (control-$e^{-iHt}$) and consumed Bell pairs for the cases in Fig.~\ref{three-dev}. The upper panels depict the total cost, while the lower panels present the cost normalized by the number of output states.}
    \label{fig:cost}
\end{figure}

In this appendix, we discuss the resource costs associated with the distributed setting. In our distributed filtering protocol, non-local operations are required. specifically, a control-SWAP gate per iteration that acts across the single auxiliary qubits on each device, facilitated by shared Bell pairs.
Also, due to the overhead introduced by postselection, additional resource consumption arises, such as an increased number of controlled time evolutions and greater use of Bell pairs.

We begin by emphasizing that the single-device filter requires only one control-$e^{-iHt}$ operation per iteration and no Bell pairs. Moreover, since there is no postselection involved in the single-device protocol, it always succeeds. Consequently, the total number of control-$e^{-iHt}$ operations grows linearly with the number of iterations $k$.

In contrast, for the multi-device filter, each party must apply a local control-$e^{-iHt}$ operation per iteration. Thus, the total number of control time evolutions per iteration equals the number of participating parties. Furthermore, due to the presence of postselection, the protocol may need to restart multiple times, leading to additional resource costs.

In addition to time evolutions, a shared Bell state $\ket{\Phi^+} = \frac{1}{\sqrt{2}}(\ket{00} + \ket{11})$ is required in each iteration to teleport Bob’s auxiliary qubit to Alice via local operations and classical communication. Once the teleportation is complete, Alice can apply the control-SWAP gate locally and determine whether the measurement outcomes satisfy the postselection criterion. If the protocol proceeds to the next iteration, the local auxiliary qubit registers can be reset and reused. The detailed circuit implementation is shown in Fig.~\ref{fig:teleport}.
This teleportation-based approach generalizes straightforwardly to the $s$-device filter, which requires a total of $s - 1$ Bell pairs per iteration to teleport all auxiliary qubits to Alice, allowing her to perform the control permutations locally.

These resource requirements are closely tied to the cumulative success rate of the protocol. Since each iteration carries a probability of failure and possible restart, resource wastage can accumulate. Using Monte Carlo method, we compute the corresponding resource costs, as illustrated in Fig.~\ref{fig:cost}.
In the \emph{weak} postselection regime, the cumulative success rate plateaus after a few iterations, implying that the protocol will almost always succeed beyond a certain $k$. As a result, the total resource cost ultimately scales linearly with $k$.
However, in the \emph{strong} postselection regime, the cumulative success rate decays exponentially, meaning that the protocol continues to have a non-negligible failure probability at each iteration. Consequently, the total resource cost grows exponentially with $k$.
Moreover, we illustrate the cost normalized by the number of output states as shown in the lower panels of Fig.~\ref{fig:cost}), since the distributed protocols can yield multiple identical eigenstates. 
They reveal the per-state resource requirements, enabling a fairer comparison among different settings.

We also note that when the controlled-SWAP operation is affected by noise, such as from an imperfect shared Bell pair or a noisy teleportation process, the protocol still functions, though with reduced performance depending on the noise level (cf. Fig. 12 in~\cite{schiffer2025quantum}).

\section{Additional numerical experiments~\label{appx-complete}}

In the final subsection of this Appendix, we present additional numerical experiments on several product state instances of the form $\ket{p(\theta)}=(\cos\theta\ket{0}+\sin\theta\ket{1})^{\otimes n}$, where we choose $\theta\in\{-\pi/4,\pi/4\}$ and $n\in\{4,5,6\}$.
We show one standard deviation ($\pm\sigma_{\mathrm{data}}$) of the data around the mean value in Fig.~\ref{2-dev-EV-appendix}, \ref{2-dev-rounds-appendix},~\ref{3-dev-appendix} and \ref{3-dev-rounds-appendix}. 
Different from the numerics in the main text, where the average values are computed as accurately as possible, here we fix the number of repetition times to $10^6$ for Fig.~\ref{2-dev-EV-appendix} and~\ref{2-dev-rounds-appendix} and $10^5$ for Fig.~\ref{3-dev-appendix} and~\ref{3-dev-rounds-appendix}. 
For each iteration $k$, we compute the average value from the data that meet the postselection criteria, and discard the others. 
Then, especially for \emph{strong} postselection, the data is computed from smaller ensembles,  leading to noticeable statistical fluctuations at large $k$.
Also, we note that due to the randomness inherent in the filtering process, the error $|\sigma_{\mathrm{data}}|$ can occasionally be large. 
To clearly illustrate the decrease in average variance (not the eigenvalue spreading) across different iterations, the y-axis is plotted on a logarithmic scale. 
As a result, the estimates for the lower error bars can extend below zero, causing the shaded error regions to cover all the areas below the data points on the plots.

\begin{figure}
    \centering
    \includegraphics[width=0.8\linewidth]{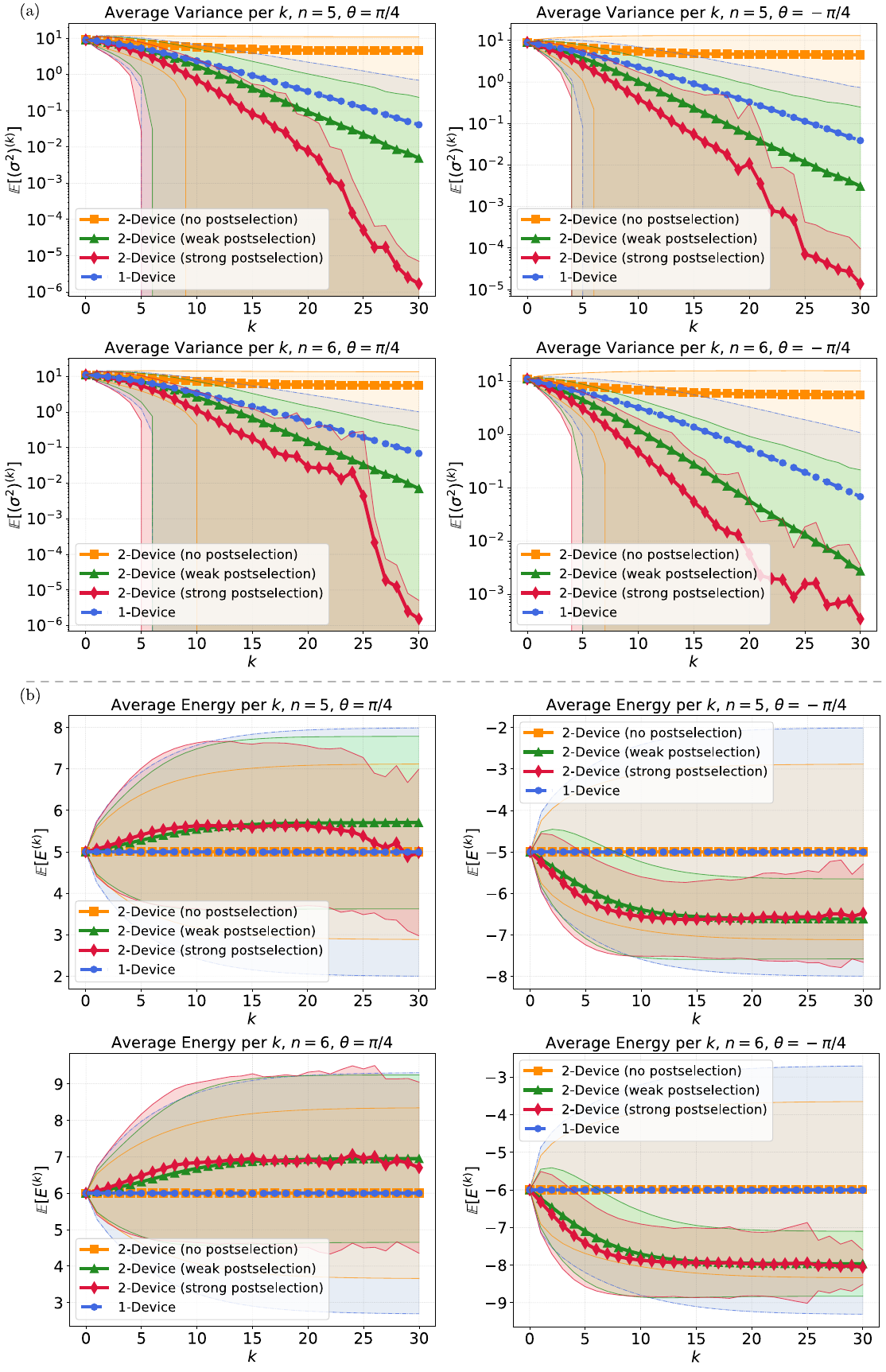}
    \caption{Average variance and energy of the pure state ensemble per iteration $k$. Four cases are shown including the instances in the main text: $n=5$ or $6$; $\theta=\pm\pi/4$. The algorithm is repeated $10^6$ times, with each repetition lasting up to 30 iterations. The fluctuations in the average variance for the \emph{strong} postselection case at large $k$ are attributed to the limited data available from the exponentially decreasing number of surviving rounds. }
    \label{2-dev-EV-appendix}
\end{figure}

\begin{figure}
    \centering
    \includegraphics[width=0.8\linewidth]{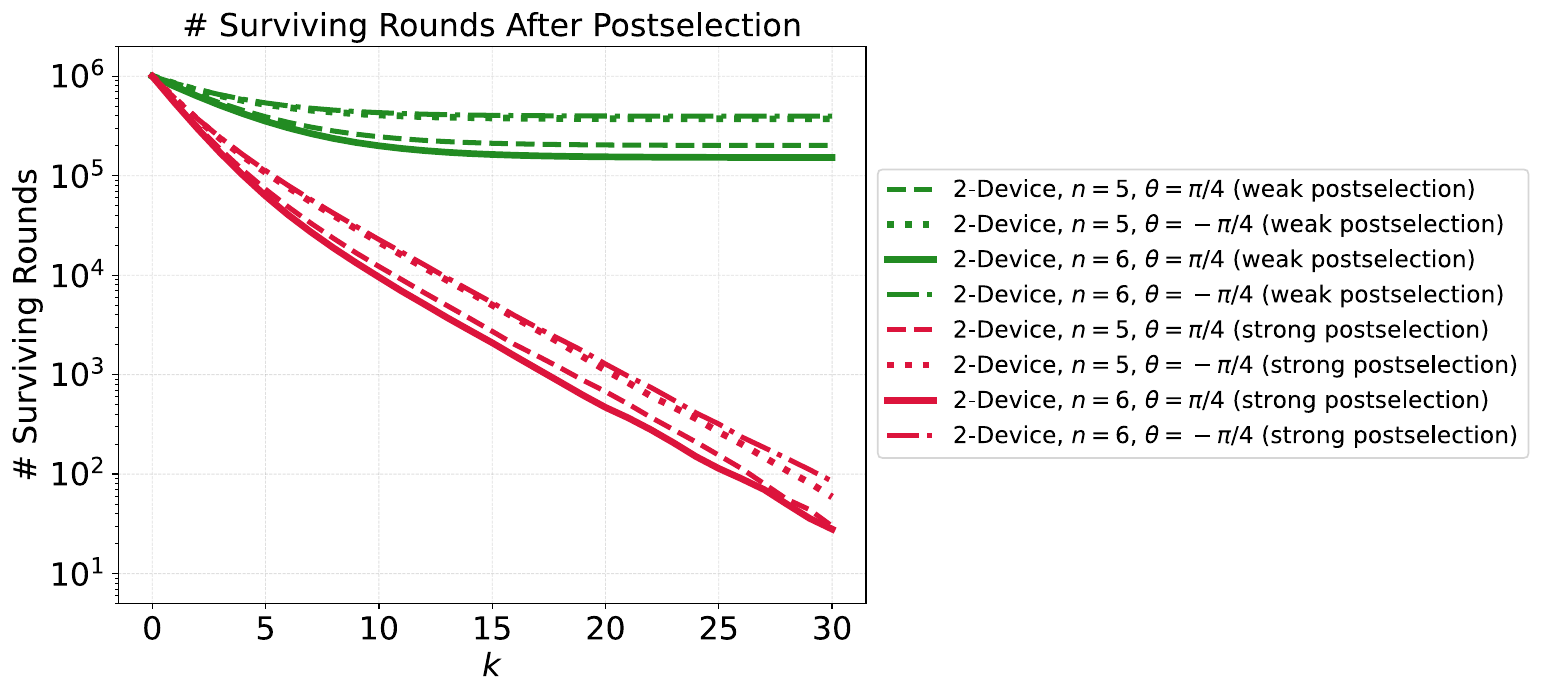}
    \caption{The surviving rounds of repetitions after postselection at each iteration, for all four numerical instances in Fig.~\ref{2-dev-EV-appendix}.}
    \label{2-dev-rounds-appendix}
\end{figure}

\begin{figure}
    \centering
    \includegraphics[width=1.0\linewidth]{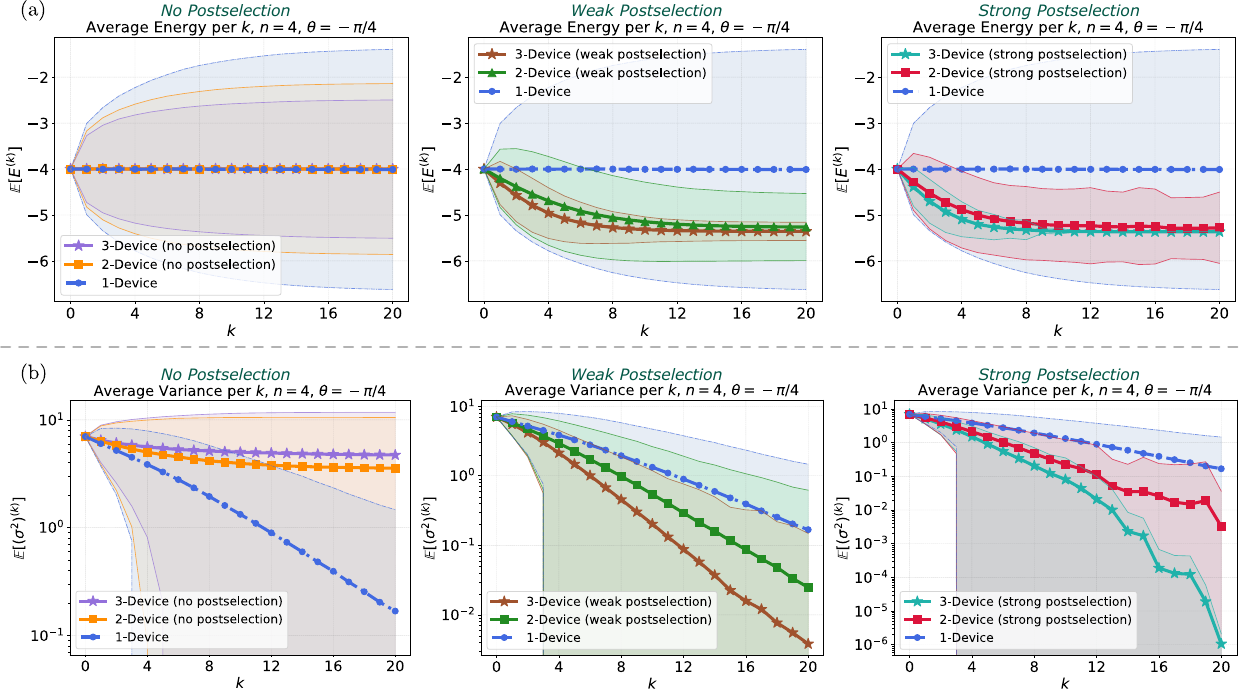}
    \caption{The numerical experiment results with the same inputs in Fig.~\ref{three-dev} (comparing single-, two- and three-device algorithms with $n=4$, $\theta=-\pi/4$) but with fixed repetition numbers ($10^5$) of the algorithm at the beginning.}
    \label{3-dev-appendix}
\end{figure}

\begin{figure}
    \centering
    \includegraphics[width=0.7\linewidth]{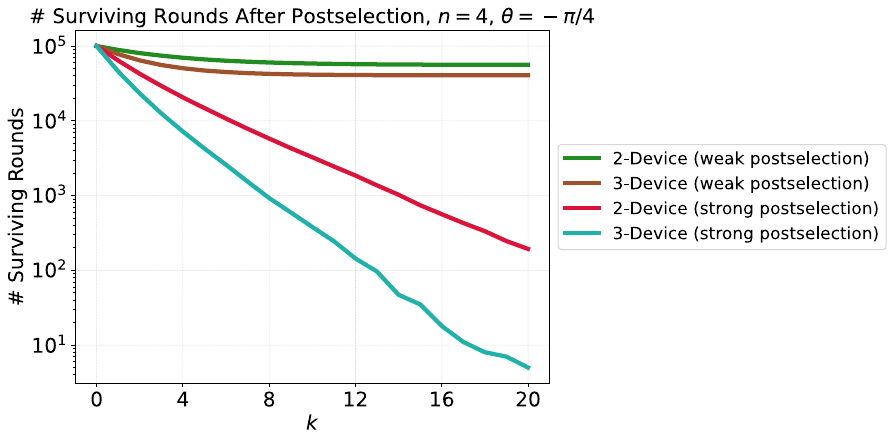}
    \caption{The surviving rounds of repetitions after postselection at each iteration for the numerical instances in Fig.~\ref{3-dev-appendix}.}
    \label{3-dev-rounds-appendix}
\end{figure}

\end{document}